\documentclass[12pt]{article}
\usepackage{newinutile}
\usepackage{latexsym,amsmath,amsfonts,amsthm,amssymb}
\usepackage{epsfig,graphics,color,calc,graphicx,pict2e}
\usepackage{multirow}
\usepackage{mathrsfs}
\usepackage{tikz}
\usepackage{tikz}
\usetikzlibrary{shapes,arrows}
\usetikzlibrary{intersections,matrix, positioning}

\usepackage{subfigure} 
\newtheorem{theorem}{Theorem}

\newcommand{\newatop}[2]{\genfrac{}{}{0pt}{}{#1}{#2}}

\newlength{\pecettawidth}
\begin{document}
\title{Phase transitions in random mixtures of elementary cellular automata}

\author{Emilio N.M.\ Cirillo}
\email{emilio.cirillo@uniroma1.it}
\affiliation{Dipartimento di Scienze di Base e Applicate per l'Ingegneria, 
             Sapienza Universit\`a di Roma, 
             via A.\ Scarpa 16, I--00161, Roma, Italy.}

\author{Francesca R.\ Nardi}
\affiliation{Dipartimento di Matematica e Informatica ``Ulisse Dini'', 
viale Morgagni 67/a, 50134, Firenze, Italy.}
\email{francescaromana.nardi@unifi.it}

\author{Cristian Spitoni}
\affiliation{Institute of Mathematics,
University of Utrecht, Budapestlaan 6, 3584 CD Utrecht, The~Netherlands.}
\email{C.Spitoni@uu.nl}


\begin{abstract}
We investigate one--dimensional Probabilistic Cellular Automata, 
called Diploid Elementary Cellular Automata (DECA), obtained as 
random mixture of two different Elementary Cellular Automata rules.
All the cells are updated synchronously and 
the probability for one cell to be $0$ or $1$ at time $t$ depends 
only on the value of the same cell and that of its neighbors at time $t-1$.
These very simple models show a very rich behavior strongly 
depending on the choice of the two Elementary Cellular Automata 
that are randomly mixed together and on the parameter which 
governs probabilistically the mixture. 
In particular, 
we study the existence of phase transition for the whole set of possible 
DECA obtained by mixing the null rule which associates $0$ to any 
possible local configuration, with any of the other $255$ 
elementary rule.
We approach the problem analytically via a Mean Field 
approximation and via the use of a rigorous approach based 
on the application of the Dobrushin Criterion. 
The distinguishing trait of our result is the 
possibility to describe the behavior of the whole set 
of considered DECA without exploiting the local properties 
of the individual models. 
The results that we find are coherent with numerical 
studies already published in the scientific literature and 
also with some rigorous results proven for some specific models. 
\end{abstract}





\vfill\eject

\maketitle

\section{Introduction}
\label{s:intro} 
\par\noindent

Probabilistic Cellular Automata (PCA) generalize deterministic 
Cellular Automata (CA) as discrete--time Markov chains. 
Despite the simplicity of their stochastic evolution rules,  
PCA exhibit a large variety of dynamical behaviors and 
for this reason are powerful modeling tools (see \cite{LN2016} for 
a general introduction to the topic).
In this paper 
we study the  relaxation towards stationarity of a family of  
one--dimensional PCA, 
called \emph{Diploid} Elementary Cellular Automata (DECA),
which are defined as Bernoulli mixtures 
of two Elementary Cellular Automata (ECA) rules \cite{W83,W84}.
DECA have originally introduced and analyzed in \cite{Fautomata2017} 
by means of numerical simulations.

By varying the ECA chosen in the mixture, the class of DECA considered 
in the present manuscript is indeed very rich an includes among the others: 
the \emph{percolation} PCA studied in  \cite{BMM2013} and \cite{T2004}, 
the  \emph{noisy additive} PCA \cite{MM2014}, the \emph{Stavskaya's} PCA \cite{Me2011}
and the \emph{directed animals} PCA \cite{Dh83}.

The long--time limit of the PCA has been the subject of many 
numerical and theoretical results in the last fifty years,  
see for instance \cite{T1990,T2004}. 
In this paper we focus on the properties of DECA's stationary states 
in function of the parameter $\lambda$ governing the Bernoulli mixture. 
In particular, we study the presence of \emph{phase transitions} associated 
to multiple invariant measures. 

In case of uniqueness of the invariant 
measure, a natural questions are related to attractiveness and
ergodicity of the system \cite{T2004}.
However, ergodicity will not be the focus of this paper and we just recall 
that the uniqueness of the invariant measure does not imply 
ergodicity \cite{CM2010,JK2015}. 
We will be interested instead to
the structure of the \emph{phase diagram} of the DECA in relation to 
the mixing parameter $\lambda$ and to the
choice of the two mixed ECA.

A rigorous study of the phase diagram of DECA is possible only 
for a tiny subset of the ECA rules. For this reason, we thus use 
a \emph{Mean Field} (MF) approximation \cite{Gv87,M90} to get a wide 
overview of the possible behaviors of all the possible DECA. 
The MF approximation assumes that at a given time the values of the cells 
are independent and not correlated with each other. Thus, the joint 
probability of the
neighborhood state is a product of the single--site probabilities. 
Therefore, a polynomial  on these single--site probabilities is derived
and its curve can be used to classify the DECA, in terms of the presence 
of phase  transition \cite{M90}. By the MF approximation we are able to explain the presence of the phase transitions suggested by the simulations in 
\cite{Fautomata2017}.  

Moreover, we can provide rigorous lower bounds for the critical point,
by using a Dobrushin single site sufficient condition 
\cite{Dpit1971}, stated in the case 
of PCA and extended in \cite{MScmp1991}. This \emph{Dobrushin criterion}
provides an instrument to prove ergodicity, and hence  
existence of a unique invariant measure, to be compared with the results 
of the MF approximation.

A third contribution of the present paper is the description of the 
relaxation towards stationarity in the finite volume regime. By looking at 
the DECA from the perspective of a \emph{finite volume} Markov chain, we 
show that for any finite size $n$ of the chain and for the mixing 
parameter $\lambda$ large enough, the system has essentially two time scales, 
sharing some features with PCA which exhibit metastable states
\cite{CNS2016}. 
On a small time scale, the chain seems to be frozen in a non--null 
stationary state (i.e., with a non--null asymptotic density), while on 
an exponentially larger time the system relaxes abruptly to the 
unique stationary configuration with zero density. 

The paper is organized as follows. In Section~\ref{s:modello} we define 
the class of DECA of interest and recall the findings of \cite{Fautomata2017}. 
In Section~\ref{s:spiegazione} we introduce the MF model and prove the 
uniqueness of the invariant measure in case of \emph{odd} models and the 
presence of a phase transition for a subset of \emph{even} models, 
for $\lambda$ large enough. In Section~\ref{s:infinite} we find a
lower bound for the critical parameter $\lambda_\textup{c}$ by 
using a \emph{Dobrushin criterion} and we prove that  
for $\lambda<1/3$ the Dobrushin criterion ensures that the invariant 
measure is unique in the infinite volume case and coincides with the 
delta measure in $0$. Furthermore, according to the number 
of \emph{marginal cells} of the neighborhood of the local rule, 
we improve this lower bound for a subclass of models. 
In Section~\ref{s:finito} we consider the DECA in finite volume.
In this regime we prove the convergence of the system towards 
the stationary state $0$ with probability one. Moreover, we show 
a behavior resembling metastability, namely, the persistence in a non--null 
state for an exponentially long time before an abrupt transition 
towards the state $0$.

\section{Phase transitions in diploid elementary cellular automata}
\label{s:modello} 
\par\noindent
A finite cellular automata with binary states and periodic boundary 
condition is defined considering a \emph{set of states} $Q=\{0,1\}$ and a 
\emph{ring} $\mathbb{L}_n=\mathbb{Z}/n\mathbb{Z}$ made of $n$ cells. 
The \emph{configuration space} is 
$X_n=Q^{\mathbb{L}_n}$.
For $x\in X_n$, $x_i$ is called \emph{value} of the cell $i$ or 
\emph{occupation number} of the cell $i$. 
The configuration with all the cell states equal to zero will be 
simply 
denoted by $0$ and, similarly, the one with all the cell states 
equal to $1$ will be denoted by $1$. 

In elementary cellular automata
all cells are updated synchronously  
so that the state of each cell is updated 
according to the state of the cell itself and to that 
of the two neighboring cells. 
The set of these three cells will be called \emph{neighborhood} 
of a given cell. 
More precisely, given a \emph{local rule} 
$f:Q^3\to Q$, 
we denote by $F:X_n\to X_n$ the map defined by letting 
\begin{displaymath}
(F(x))_i
=
f(x_{i-1}, x_{i}, x_{i+1})
\end{displaymath}
for any $i\in\mathbb{L}_n$.
The \emph{Elementary Cellular Automata} (ECA) associated with the 
local rule $f$ is the collection
of all the sequences of configurations 
$(x^t)_{t\in\mathbb{N}}$ obtained by applying the map $F$ iteratively, 
namely, such that 
$x^t=F(x^{t-1})$.
The particular sequence $(x^t)_{t\in\mathbb{N}}$ 
such that $x^0=x\in X_n$ is called \emph{trajectory} 
of the cellular automaton associated 
with the \emph{initial condition} $x$. 

Each of the possible  
$256$ local rules $f$ is identified by the integer number 
$W\in\{0,\dots,255\}$ 
such that 
\begin{equation}
\label{fin000}
\begin{array}{rcl}
W
&\!\!=&\!\!
f(1,1,1)\cdot2^7
+
f(1,1,0)\cdot2^6
+
f(1,0,1)\cdot2^5
+
f(1,0,0)\cdot2^4
\\
&&\!\!
{\displaystyle
+
f(0,1,1)\cdot2^3
+
f(0,1,0)\cdot2^2
+
f(0,0,1)\cdot2^1
+
f(0,0,0)\cdot2^0
=\sum_{i=0}^7 c_i\cdot 2^i
,
}
\end{array}
\end{equation}
where the last equality defines the coefficients $c_i$,
see Figure~\ref{f:coeff}.
The collection of the digits $c_7c_6c_5c_4c_3c_2c_1c_0$
is the binary representation of the number $W$. 
We shall often denote the ECA 
with both the decimal and the binary representation, 
namely, we shall write 
$W(c_7c_6c_5c_4c_3c_2c_1c_0)$.
Note that all the rules represented by an even number 
associated to the local configuration $000$ the states $0$. 

\begin{figure}
\begin{center}
\begin{tikzpicture}[b/.style={draw, minimum size=3mm,   
       fill=black},w/.style={draw, minimum size=3mm},
       m/.style={matrix of nodes, column sep=1pt, row sep=1pt, draw, label=below:#1}, node distance=1pt]

\matrix (A) [m=$c_7$]{
|[b]|&|[b]|&|[b]|\\
&\\
};
\matrix (B) [m=$c_6$, right=of A]{
|[b]|&|[b]|&|[w]|\\
&\\
};
\matrix (C) [m=$c_5$, right=of B]{
|[b]|&|[w]|&|[b]|\\
&\\
};
\matrix (D) [m=$c_4$, right=of C]{
|[b]|&|[w]|&|[w]|\\
&\\
};
\matrix (E) [m=$c_3$, right=of D]{
|[w]|&|[b]|&|[b]|\\
&\\
};
\matrix (F) [m=$c_2$, right=of E]{
|[w]|&|[b]|&|[w]|\\
&\\
};
\matrix (G) [m=$c_1$, right=of F]{
|[w]|&|[w]|&|[b]|\\
&\\
};
\matrix (H) [m=$c_0$, right=of G]{
|[w]|&|[w]|&|[w]|\\
&\\
};
\end{tikzpicture}
\end{center}
\vskip -0.5 cm
\caption{Schematic representation of the coefficients $c_i$: 
$c_i$ is equal to one if the cell value corresponding to the associated 
configuration of the local neighborhood is $1$; otherwise it is zero. 
In the picture, black squares represents ones and empty squares 
represent zeroes.}
\label{f:coeff}
\end{figure}
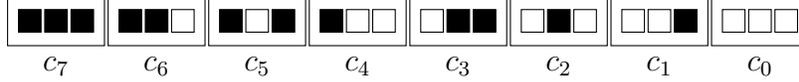

Some examples.\ 
The rule $0$ is called the \emph{null} rule 
and associates the state $0$ to any configuration 
in the neighborhood. 
The rule $22(00010110)$ 
associates the state $0$ to any configuration 
in the neighborhood 
but for the three local configurations in which one single $1$ is present 
in the neighborhood ($001$, $010$, and $100$) to which it associates $1$. 
The rule $150(10010110)$ 
associates the state $0$ to any configuration 
in the neighborhood 
but for the four local configurations in which an odd number of $1$'s 
is present in the neighborhood ($001$, $010$, $100$, and $111$) 
to which it associates $1$. 
The rule $204(11001100)$ is called the \emph{identity} 
and associates to any configuration 
in the neighborhood the state of the cell at the center (namely, the 
cell that one is going to update). 
The rule $224(11100000)$ 
associates the state $1$ to any configuration 
in the neighborhood 
but for the local configuration $000$ to which it associates $0$. 
The rule $232(11101000)$ is called the \emph{majority rule} 
and associates to any configuration 
in the neighborhood the majority state, namely $0$ to 
$000$, $001$, $010$, and $100$ and $1$ to the others.
The rule $255(11111111)$ 
associates the state $1$ to any configuration 
in the neighborhood. 

In this context a Probabilistic Cellular Automata, 
called \emph{probabilistic} or \emph{stochastic} ECA, 
is a Markov chain $(\xi^t)_{t\in\mathbb{N}}$ on the 
configuration space $X_n$ with transition matrix 
\begin{equation}
\label{mod000}
p(x,y)
=
\prod_{i\in\mathbb{L}_n}
p_i(y_i|x)
\;\textrm{ with }\;
p_i(y_i|x)
=
y_i\phi(x_{i-1},x_i,x_{i+1})
+
(1-y_i)[1-\phi(x_{i-1},x_i,x_{i+1})]
\end{equation}
where $\phi:Q^3\to[0,1]$ has to be interpreted as the probability 
to set the cell to $1$ given the neighborhood 
$x_{i-1}x_ix_{i+1}$ and, similarly, 
$1-\phi$ the probability to select $0$.
We denote by $P_x$ the probability associated with the 
process started at $x\in X_n$. 
We shall denote by $\mu^x_t(y)=P_x(\xi^t=y)$ the probability that the 
chain started at $x$ will be in the configuration $y$ at time 
$t$. 
Abusing the notation, 
$\mu^x_t(Y)=P_x(\xi^t\in Y)$ will denote the probability that the 
chain started at $x$ will be in the set of configurations $Y\subset X_n$ 
at time $t$. 

An important class of stochastic ECA is made of those models 
obtained by randomly mixing two of the $256$ elementary cellular 
automata. More precisely, given $\lambda\in(0,1)$ and 
picked two local rules $f_1\neq f_2$, the stochastic ECA 
defined by 
\begin{equation}
\label{mod010}
\phi=(1-\lambda) f_1+\lambda f_2
\end{equation} 
is called a \emph{Diploid} ECA (DECA). 
Note that in the limiting cases $\lambda=0,1$ or $f_1=f_2$ a
(deterministic) ECA is recovered. 

It is important to note that the time evolution of the diploid ECA 
can be described as follows: at time $t$ for each 
cell $i\in\mathbb{L}_n$ one chooses either the rule $f_1$ with 
probability $1-\lambda$ or the rule $f_2$ with probability 
$\lambda$ and performs the updating based on the neighborhood configuration 
at time $t-1$. 
Indeed, with this algorithm the 
probability to set the cell to $1$ a time $t$ is 
$0$ if $f_1=0$ and $f_2=0$ (where the local rules are computed 
in the neighborhood configuration at time $t-1$), 
$1-\lambda$ if $f_1=1$ and $f_2=0$, 
$\lambda$ if $f_1=0$ and $f_2=1$, 
$1$ if $f_1=1$ and $f_2=1$, which is coherent with the definition 
\eqref{mod010}

In the following we shall consider the case 
in which $f_1$ is the null rule (i.e., ECA 0) 
and $f_2$ is any other rule; those diploid elementary cellular 
automata will be called NDECA. 
In order to further simplify the exposition, we will call NDECA $n$ the NDECA with $f_2$ the ECA $n$.
We note that
the measure concentrated on the zero configuration $0$ is an
invariant measure for the finite volume NDECA in case in which the 
$f_2$ rule is even.

In this framework the main question is to understand if in the 
infinite volume limit, namely, $n\to\infty$, a different 
stationary measure exists, with a positive value of the 
average cell occupation number. 

This problem has been extensively studied in 
\cite{Fautomata2017} via numerical simulations: 
the diploid is started at
an initial configuration 
$x\in X_n$ 
in which cells are populated 
with zeros or ones with equal probability. 
For the chain $\xi^t$ the quantity 
$P_x(\xi_i^t=1)$
is the average value of the cell $i$ at time $t$; 
its spatial average 
\begin{equation}
\label{e:density}
\delta_x(t)=\frac{1}{n}\sum_{i\in \mathbb{L}_n } P_x(\xi_i^t=1)
\end{equation}
is called \emph{density} and represents the quantity of 
interest in these simulations. 
In particular, NDECA with 
$n=10^4$ cells have been extensively simulated 
for the time $T=5\cdot10^3$;
the fraction of cells set to $1$ measured in the final configuration 
has been considered as the stationary measure of the density.  
Clearly, whether or not  this number is an estimate of the averaged density 
along 
an infinite long run of the diploid in the infinite volume limit 
$n\to\infty$, it will depends on the infinite volume ergodic properties 
of the chain. 
The measure is repeated for any choice of the elementary 
rule $f_2$ and for many different 
choices of the mixing rate $\lambda\in(0,1)$. 
As reported in 
\cite[Table~1 and Figure~1]{Fautomata2017},
if the rule $f_2$ is chosen from the list 
\begin{eqnarray*}
\mathcal{F}
&\!\!\!\!\!\!\!\!=&
\!\!\!\!\!\!\!\!
\{
18, 22, 26, 28, 30, 50, 54, 58, 60, 62, 78, 90, 94, 110,
122, 126, 146, 150, 154, 156, 158, \\
&\phantom{=\{}&
\!\!\!\!\!
178, 182, 186, 188,
190, 202, 206, 218, 220, 234,238, 250, 254\}
\end{eqnarray*}
a continuous transition is observed, in the sense that 
the measured stationary density is equal to zero 
for $\lambda\in(0,\lambda_\textup{c})$ and 
is \emph{continuously} monotonically growing for $\lambda\ge\lambda_\textup{c}$. 
The critical value $\lambda_\textup{c}$ is close to $0.7$ but it seems to 
depend on the choice of the rule $f_2$,  
see Figures~\ref{f:18} and \ref{f:254}. 

These results are partially explained in the following sections by means of a 
MF approximation and by using rigorous arguments based in 
the Dobrushin single site condition.

Our general analysis will cover models well known in the literature, whose asymptotic behavior is studied rigorously and/or  numerically. We will
consider indeed \emph{directed animals PCA} (NDECA 17), \emph{diffusion PCA} (NDECA 18),  \emph{noisy additive PCA} (NDECA 102), 
the \emph{Stavskaya model} (NDECA 238), and the \emph{percolation PCA} (NDECA 254).

\section{Mean field approximation}
\label{s:spiegazione} 
\par\noindent
We derive a \emph{Mean Field} (MF) approximation of 
the stationary density of any NDECA and find results 
consistent with the numerical predictions in \cite{Fautomata2017}.

Since $f_1$ is the null rule, from 
(\ref{mod000}) and (\ref{mod010}), we have that, 
for any $i=1,\dots.n$, 
$p_i(1|y)=\lambda \textbf{1}\{f_2(y_{i-1}y_iy_{i+1})=1\}$.
Thus, 
\begin{equation}
\label{eq:it010}
\begin{array}{rcl}
P_x(\xi_i^t=1) 
&\!\!=&\!\!
{\displaystyle
\sum_{y\in X_n}P_x(\xi_i^t=1|\xi^{t-1}=y)  {P}_x( \xi^{t-1}=y) 
}
\\
&\!\!=&\!\!
{\displaystyle
\sum_{y\in X_n} p_i(1|\xi^{t-1}=y)  P_x(\xi^{t-1}=y) 
}
\\
&\!\!=&\!\!
{\displaystyle
\lambda\sum_{y\in X_n} \textbf{1}\left\{f(y_{i-1},y_i,y_{i+1})=1 
        \right\} P_x(\xi^{t-1}=y) 
}
\\
&\!\!=&\!\!
{\displaystyle
\lambda \sum_{z\in f_2^{-1}(1)} 
P_x((\xi^{t-1}_{i-1},\xi^{t-1}_i,\xi^{t-1}_{i+1})=z)
}
,
\end{array}
\end{equation}
where, as usual, $f_2^{-1}(1)$ is the counter image of $1$ under $f_2$, 
namely, the set of neighbors (i.e., triples) mapped to one by the local 
rule $f_2$.

Considering a MF approximation, here, 
means approximating $P_x((\xi^{t-1}_{i-1},\xi^{t-1}_i,\xi^{t-1}_{i+1})=z)$ 
with the product 
$P_x((\xi^{t-1}_{i-1}=z_1)P_x((\xi^{t-1}_{i-1}=z_2)P_x((\xi^{t-1}_{i-1}=z_3)$, 
where $z=(z_1,z_2,z_3)$.
Thus, if we let $a_i(t;x)$ to be the MF approximation of the probability 
that the value of the cell $i$ is one at time $t$, from 
\eqref{eq:it010} we get 
\begin{equation}
\label{e:mean_general}
a_i(t;x)
=
\lambda 
\sum_{(z_1,z_2,z_3)\in f_2^{-1}(1)}
\prod_{k=1}^3
[z_k a_{i-2+k}(t-1;x)+(1-z_k)(1-a_{i-2+k}(t-1;x))]
,
\end{equation}
which is the MF iterative equation for the occupation probability.

Starting from an homogeneous initial configuration $x$, which is 
the case in the simulations performed in \cite{Fautomata2017}, 
the MF iterations \eqref{e:mean_general} preserve such a 
homogeneity character. 
Thus, we seek for the NDECA phases by looking for homogeneous 
stationary (not dependent on time) solution $a$ of the 
MF equation \eqref{e:mean_general}, that is to say, we consider the 
equation 
\begin{equation}
\label{e:MF_stat}
a
=
\lambda 
\sum_{(z_1,z_2,z_3)\in f_2^{-1}(1)}
\prod_{k=1}^3
[z_k a+(1-z_k)(1-a)]
.
\end{equation}

If the ECA $f_2$ is represented by the binary sequence of digits 
$c_7 c_6 c_5 c_4 c_3 c_2 c_1 c_0$, see \eqref{fin000} and 
Figure~\ref{f:coeff}, 
then \eqref{e:MF_stat} becomes
\begin{displaymath}
a
=
\lambda 
[c_7 a^3
 +(c_6+c_5+c_3)a^2(1-a)
 +(c_4+c_2+c_1)a(1-a)^2
 +c_0(1-a)^3]
,
\end{displaymath}
which can be rewritten as 
\begin{equation}
\label{e:a_general}
a
=
\lambda [(c_7-S_2+S_1-c_0)a^3+(S_2-2S_1+3c_0)a^2+(S_1-3c_0)a+c_0]
,
\end{equation}
where $S_2=c_6+c_5+c_3$ 
is the number of configurations in $f_2^{-1}(1)$ in which only two cells 
have value one 
and 
$S_1=c_4+c_2+c_1$ 
is the number of configurations in $f_2^{-1}(1)$ in which one single cell 
has value one.

\setlength{\unitlength}{1.3pt}
\begin{figure}[t]
\begin{picture}(400,70)(-120,-20)
\thinlines
\put(-30,0){\vector(1,0){70}}
\put(0,-30){\vector(0,1){70}}
\put(41,0){${\scriptstyle a}$}
\put(1,40){${\scriptstyle q(a)}$}
\put(0,20){\circle*{2}}
\put(20,-10){\circle*{2}}
\qbezier(-20,-20)(-10,50)(0,20)
\qbezier(0,20)(15,-25)(20,-10)
\qbezier(20,-10)(25,5)(30,40)
\put(70,0){\vector(1,0){70}}
\put(100,-30){\vector(0,1){70}}
\put(141,0){${\scriptstyle a}$}
\put(101,40){${\scriptstyle q(a)}$}
\put(100,20){\circle*{2}}
\put(120,-10){\circle*{2}}
\qbezier(80,-20)(85,50)(100,20)
\qbezier(100,20)(115,-10)(127,-15)
\qbezier(127,-15)(135,-17)(140,40)
\end{picture}
\vskip .5 cm
\caption{Two possible graphs of the cubic $q(a)$ in the case 
$c_7-S_2+S_1-1>0$.}
\label{fig:odd}
\end{figure}
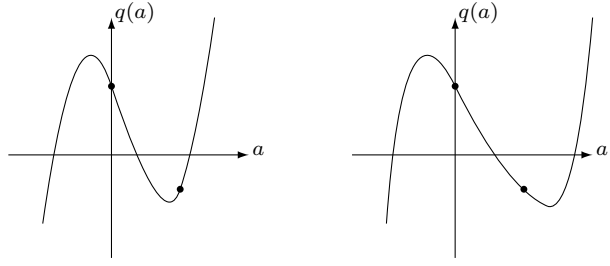

\subsection{Odd NDECA}
\label{s:odd} 
\par\noindent
We say that a NDECA is odd if $c_0=1$:
the ECA $f_2$ maps the configuration $000$ to one. 
In this case we prove that 
\eqref{e:a_general} has a unique solution in $[0,1]$.

We first note that in this case the equation \eqref{e:a_general} can 
be rewritten as 
\begin{equation}
\label{e:odd000}
q(a)\equiv
(c_7-S_2+S_1-1)a^3+(S_2-2S_1+3)a^2
+\Big(S_1-3-\frac{1}{\lambda}\Big)a+1
=0
\end{equation}
and compute $q(0)=1$ and $q(1)=c_7-1/\lambda<0$.

\medskip
\par\noindent
\textit{Case} $c_7-S_2+S_1-1=0$: 
the polynomial $q(x)$ has degree equal to one or two. 
In both cases, since its graph in the plane $a$--$q(a)$
has to pass through the points $(0,1)$ and 
$(1,c_7-1/\lambda)$, with $c_7-1/\lambda<0$, we have that 
such a graph intersects the segment $[0,1]$ in one single  
point. 

\medskip
\par\noindent
\textit{Case} $c_7-S_2+S_1-1>0$: using that 
$\lim_{a\to\pm\infty}q(a)=\pm\infty$
and, again,  
the fact that 
the graph of the cubic polynomial $q(a)$ in the plane $a$--$q(a)$
has to pass through the points $(0,1)$ and 
$(1,c_7-1/\lambda)$, with $c_7-1/\lambda<0$, we have that 
such a graph intersects the segment $[0,1]$ in one single  
point (see Figure~\ref{fig:odd}, where two possible 
situations are depicted). 

\medskip
\par\noindent
\textit{Case} $c_7-S_2+S_1-1<0$: we note that
$\lim_{a\to\pm\infty}q(a)=\mp\infty$, 
recall $q(0)=1$ and $q(1)=c_7-1/\lambda<0$,
and
compute 
$q'(a)=3(c_7-S_2+S_1-1)a^2+2(S_2-2S_1+3)a+(S_1-3-1/\lambda)$.

If $c_7=0$, the reduced discriminant 
$\Delta/4=-3S_2+S_2^2-S_2S_1+S_1^2+3(-S_2+S_1-1)/\lambda$
of the equation 
$q'(a)=0$ is negative since
the condition $-S_2+S_1-1<0$ implies $S_1\le S_2\le3$.
Thus, $q'(a)$ is negative and hence
the graph of $q(a)$ intersects the segment $[0,1]$ in one single  
point.

If $c_7=1$, 
the condition $-S_2+S_1<0$ implies $S_2\ge S_1+1$.
We thus compute the reduced discriminant for all the possible 
cases 
$(S_1,S_2)=(0,1),(0,2),(0,3),(1,2),(1,3),(2,3)$
and, respectively, find
$\Delta/4=7-3/\lambda,7-6/\lambda,9-9/\lambda,
3-3/\lambda,4-6/\lambda,1-3/\lambda$.
In the last four cases the reduced discriminant is negative, 
thus, $q'(a)$ is negative and hence
the graph of $q(a)$ intersects the segment $[0,1]$ in one single  
point.

We are left with two cases for which we compute explicitly 
the solutions of the equation $q'(a)=0$.
In the case $S_1=0$ and $S_2=1$ we find 
$a_\pm=(4\pm\sqrt{7-3/\lambda})/3$: when the 
two solutions are real we have 
$0<a_-<1$
and 
$a_+>1$, hence 
the graph of $q(a)$ intersects the segment $[0,1]$ in one single point.
In the case $S_1=0$ and $S_2=2$ we find 
$a_\pm=(5\pm\sqrt{7-6/\lambda})/6$: when the 
two solutions are real they are such that
$0<a_\pm<1$, but the value of the function at the maximum 
point is negative, namely, $q(a_+)<0$. Thus,
the graph of $q(a)$ intersects the segment $[0,1]$ in one single point.

\begin{table}
\begin{center}
\begin{tabular}{c|c|c|c|c|r}
\hline\hline
$c_7$ & $S_2$ & $S_1$ & $\lambda^*$ & $a^*(\lambda)$ & decimal and binary code 
\\
\hline\hline
0 & 2 & 2 & 
$\frac{1}{2}$ &
$1-\frac{1}{2\lambda}$ &
46(00101110),
58[00111010],
60[00111100],
\\
&&&&&
78[01001110],
90[01011010],
92(01011100),
\\
&&&&&
102(01100110),
114(01110010),
\\
&&&&&
116(01110100)
\\
\hline
0 & 3 & 3 & 
$\frac{1}{3}$ &
$1-\frac{1}{3\lambda}$ &
126[01111110]
\\
\hline
1 & 3 & 2 & 
$\frac{1}{2}$ &
$2-\frac{1}{\lambda}$ &
238[11101110],
250[11111010],
252(11111100)
\\
\hline\hline
\end{tabular}
\end{center}
\caption{Case $c_7-S_2+S_1=0$: list of models for which the MF equation 
suggests existence of phase transition.
Values of $c_7$, $S_2$, and $S_1$ in the first three columns, 
critical value $\lambda^*$, 
not trivial solution $a^*(\lambda)$ (order parameter) 
of the MF equation 
\eqref{e:a_general},
decimal and binary representation of the code in the last column.
The decimal code is in square bracket if the simulation 
in \cite{Fautomata2017} shows the phase transition. 
}
\label{t:g1}
\end{table}

\subsection{Even NDECA}
\label{s:even} 
\par\noindent
We say that a NDECA is even if $c_0=0$:
the ECA $f_2$ maps the configuration $000$ to zero. 
For some of the even NDECA the MF equation 
\eqref{e:a_general} has more than one solution if $\lambda$ 
is larger than a critical value $\lambda^*$, 
that is to say, in these 
cases the system exhibits a phase transition guided by the 
parameter $\lambda$. More precisely, the MF approximation 
predicts that the stationary density $a^*(\lambda)$ 
is the \emph{order parameter} describing 
this transition and is equal to zero for $\lambda<\lambda^*$ and 
positive for $\lambda>\lambda^*$.

\begin{table}
\begin{center}
\begin{tabular}{c|c|c|c|c|r}
\hline\hline
$c_7$ & $S_2$ & $S_1$ & $\lambda^*$ & $a^*(\lambda)$ & decimal and binary code 
\\
\hline\hline
0 & 0 & 2 & 
$\frac{1}{2}$ &
$1-\sqrt{\frac{1}{2\lambda}}$ &
6[00000110], 
18[00010010], 
20(00010100)
\\
\hline
0 & 0 & 3 & 
$\frac{1}{3}$ &
$1-\sqrt{\frac{1}{3\lambda}}$ &
22[00010110]
\\
\hline
0 & 1 & 2 & 
$\frac{1}{2}$ &
$\frac{3}{2}-\sqrt{\frac{1}{4}+\frac{1}{\lambda}}$ &
14(00001110),
26[00011010],
28[00011100],
\\
&&&&&
38(00100110),
50[00110010],
52(00110100),
\\
&&&&&
70(01000110),
82(01010010),
84(01010100)
\\
\hline
0 & 1 & 3 &
$\frac{1}{3}$ &
$\frac{5}{4}-\sqrt{\frac{1}{16}+\frac{1}{2\lambda}}$ &
30[00011110],
54[00110110],
86(01010110)
\\
\hline
0 & 2 & 3 & 
$\frac{1}{3}$ &
$2-\sqrt{1+\frac{1}{\lambda}}$ &
118(01110110),
94[01011110],
62[00111110]
\\
\hline
1 & 0 & 2 & 
$\frac{1}{2}$ &
$\frac{2}{3}-\sqrt{-\frac{2}{9}+\frac{1}{3\lambda}}$ &
134(10000110),
146[10010010],
148(1001010)
\\
\hline
1 & 0 & 3 & 
$\frac{1}{3}$ &
$\frac{3}{4}-\sqrt{-\frac{3}{16}+\frac{1}{4\lambda}}$ &
150[10010110]
\\
\hline
1 & 1 & 2 & 
$\frac{1}{2}$ &
$\frac{3}{4}-\sqrt{-\frac{7}{16}+\frac{1}{2\lambda}}$ &
142(10001110),
154[10011010],
156[10011100],
\\
&&&&&
166(10100110),
178[10110010],
180(10110100),
\\
&&&&&
198(11000110),
210(11010010),
212(11010100) 
\\
\hline
1 & 1 & 3 & 
$\frac{1}{3}$ &
$\frac{5}{6}-\sqrt{-\frac{11}{36}+\frac{1}{3\lambda}}$ &
158[10011110],
182[10110110],
214(11010110)
\\
\hline
1 & 2 & 2 & 
$\frac{1}{2}$ &
$1-\sqrt{-1+\frac{1}{\lambda}}$ &
174(10101110),
186[10111010],
188[10111100],
\\
&&&&&
206[11001110],
218[11011010],
220(11011100),
\\
&&&&&
230(11100110),
242(11110010),
244(11110100)
\\
\hline
1 & 2 & 3 & 
$\frac{1}{3}$ &
$1-\sqrt{-\frac{1}{2}+\frac{1}{2\lambda}}$ &
190[10111110],
222(11011110),
246(11110110)
\\
\hline
1 & 3 & 3 & 
$\frac{1}{3}$ &
$\frac{3}{2}-\sqrt{-\frac{3}{4}+\frac{1}{\lambda}}$ &
254[11111110] 
\\
\hline\hline
\end{tabular}
\end{center}
\caption{
As in Table~\ref{t:g1} for $c_7-S_2+S_1>0$.
}
\label{t:g2}
\end{table}

We first note that in this case the equation \eqref{e:a_general} can 
be rewritten as 
\begin{equation}
\label{e:even000}
\lambda p(a)\equiv
\lambda[
(c_7-S_2+S_1)a^3+(S_2-2S_1)a^2
+S_1a
]
=
a
\end{equation}
and compute 
$p(0)=0$,
$p(1)=c_7$, 
$p'(0)=S_1$,
and 
$p'(1)=3c_7-S_2$.

\medskip
\par\noindent
\textit{Case} $c_7-S_2+S_1=0$: 
if $S_2-2S_1>0$ the graph of the polynomial $\lambda p(a)$ is a convex parabola 
passing through $(0,0)$ and $(1,\lambda A)$; hence the 
equation \eqref{e:a_general} has the single solution $a=0$. 
If $S_2-2S_1=0$ the graph of the polynomial $\lambda p(a)$ is a straight 
line with slope $\lambda S_1$, 
hence the equation \eqref{e:a_general} has the single solution $a=0$. 
If $S_2-2S_1<0$ the graph of the polynomial $\lambda p(a)$ is a concave 
parabola 
passing through $(0,0)$ and $(1,\lambda A)$. 
Since $\lambda p'(a)=\lambda S_1$, the 
equation \eqref{e:a_general} has one more solution, 
besides $a=0$, provided $\lambda$ is large enough.  
The second solution appears continuously from $0$ and 
increases with $\lambda$. 
The NDECA satisfying these conditions are listed in Table~\ref{t:g1}.

\medskip
\par\noindent
\textit{Case} $c_7-S_2+S_1>0$: we note that
$\lim_{a\to\pm\infty}p(a)=\pm\infty$, 
and
recall $p(0)=0$, $\lambda p(1)=\lambda c_7$,
$p'(0)=S_2$.
The graph of the cubic polynomial $\lambda p(a)$ intersects 
the straight line $a$ for $\lambda$ sufficiently large 
if the derivative $p'(a)$ in $a=0$ is larger than $1$. Hence, 
the MF equation \eqref{e:a_general} 
has one more solution, 
besides $a=0$, provided $\lambda$ is large enough.  
The second solution appears continuously from $0$ and 
increases with $\lambda$. 
The NDECA satisfying these conditions are listed in Table~\ref{t:g2}.

\medskip
\par\noindent
\textit{Case} $c_7-S_2+S_1<0$: 
the MF equation \eqref{e:a_general} for the nine 
possible cases 
$(c_7,S_2,S_1)=
(0,1,0), (0,2,0), (0,2,1), (0,3,0), (0,3,1),
(0,3,2), (1,2,0), (1,3,0), (1,3,1)$
is solved and the 
NDECA for which a phase transition is found are listed 
in Table~\ref{t:g3}.

\begin{table}
\begin{center}
\begin{tabular}{c|c|c|c|c|r}
\hline\hline
$c_7$ & $S_2$ & $S_1$ & $\lambda^*$ & $a^*(\lambda)$ & decimal and binary code 
\\
\hline\hline
0 & 3 & 1 & 
$\frac{8}{9}$ &
$\frac{1}{4}+\sqrt{\frac{9}{16}-\frac{1}{2\lambda}}$ &
120(01111000),
108(01101100),
106(01101010)
\\
\hline
0 & 3 & 2 & 
$\frac{1}{2}$ &
$-\frac{1}{2}+\sqrt{\frac{9}{4}-\frac{1}{\lambda}}$ &
110[01101110],
122[01111010],
124(01111100)
\\
\hline
1 & 3 & 0 & 
$\frac{8}{9}$ &
$\frac{3}{4}+\sqrt{\frac{9}{16}-\frac{1}{2\lambda}}$ &
232(11101000)
\\
\hline
1 & 3 & 1 & 
$\frac{4}{5}$ &
$\frac{1}{2}+\sqrt{\frac{5}{4}-\frac{1}{\lambda}}$ &
234[11101010],
236(11101100),
248(11111000)
\\
\hline\hline
\end{tabular}
\end{center}
\caption{As in Table~\ref{t:g1} for $c_7-S_2+S_1<0$.
}
\label{t:g3}
\end{table}

\subsection{Discussion of MF results}
\label{s:dMF} 
\par\noindent
In Sections~\ref{s:odd} and \ref{s:even} and in Tables~\ref{t:g1}--\ref{t:g3}
we have provided a detailed study of the MF equation 
\eqref{e:a_general}.

We have proven that in the MF approximation odd NDECA do not 
exhibit phase transition, indeed, we have proven that equation 
\eqref{e:a_general} admits a single solution. 
This result is coherent with the numerical results discussed 
in \cite{Fautomata2017}.

In the even case, namely, $c_0=0$, the MF computation 
suggests that NDECA have to be classified through the 
parameters $c_7$, $S_2=c_6+c_5+c_3$, 
and $S_1=c_4+c_2+c_1$, where, we recall 
$S_2$ 
and 
$S_1$ 
count, respectively, 
the number of configurations in $f_2^{-1}(1)$ in which only two cells 
or only one single cell
have value one. 
Models belonging to those classes share the same behavior in the sense 
that either they all exhibit phase transition or not; moreover, 
in case of phase transition, they share both the 
critical point $\lambda^*$ and the 
order parameter $a^*(\lambda)$. 

The full 
list of rules for which the transition is found solving 
the MF equation is 
provided in Tables~\ref{t:g1}--\ref{t:g3}.
It is worth noting that the NDECA reported in 
Tables~\ref{t:g1} and \ref{t:g2}, 
namely, those for which 
$c_7-S_2+S_1\ge0$, 
exhibit a \emph{continuous} phase transition 
in the sense that at the critical point $\lambda^*$ 
the value of the order parameter is zero, 
that is to say, $a^*(\lambda^*)=0$.
On the other hand, for $c_7-S_2+S_1<0$ 
in Table~\ref{t:g3} four models are reported and the 
transition is continuous in the case 
$(c_7,S_2,S_1)=(0,3,2)$ 
whereas it is discontinuous
for 
$(c_7,S_2,S_1)=(0,3,1),(1,3,0),(1,3,1)$;
indeed, when $\lambda$ crosses the value $\lambda^*$ the stationary density 
jumps from $0$ to 
$1/4$, $3/4$, and $1/2$, respectively.
To our knowledge this is the first time in which 
discontinuous phase transitions are found for NDECA models. 

Finally, for even NDECA
the MF approximation 
predicts the existence of the phase 
transition for all the models for which the simulations 
in \cite{Fautomata2017} found the transition (see the list 
$\mathcal{F}$ given in Section~\ref{s:modello}), but for the 
NDECA with the map 202 as the rule $f_2$. 
We remark also that MF predicts the phase transition for 
many models which do not 
belong to $\mathcal{F}$.

\subsection{Examples}
\label{s:exe} 
\par\noindent
The ECA 18(00010010) is a chaotic CA belonging to Wolfram's class $W3$. 
It is also called \emph{diffusive rule}, and the reason can be understood 
by looking at the left panel in Figure~\ref{f:ECA}. 
The main feature for ECA 18 is that it creates a one at time $t$ only if there 
is a one either on its left or on its right at time $t-1$.  
Thus, it is an example of symmetric rule.

\begin{figure}[t]
\begin{center}
\includegraphics[height=4.0cm]{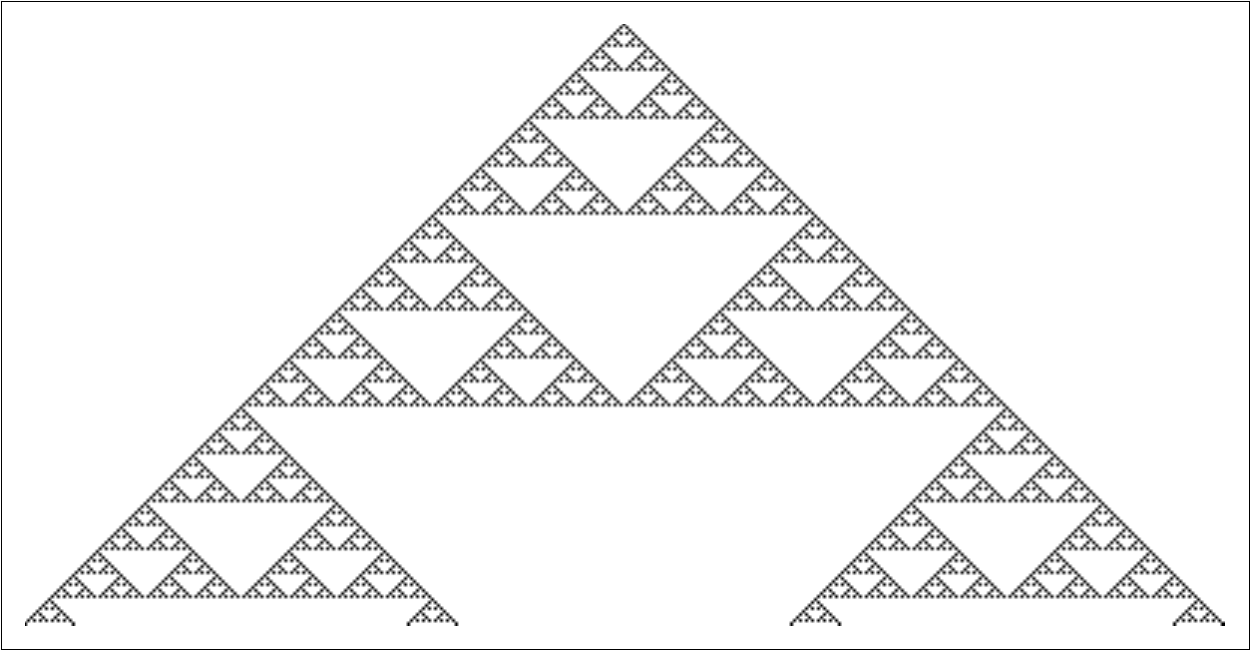}
\hskip 0.7cm
\includegraphics[height=4.0cm]{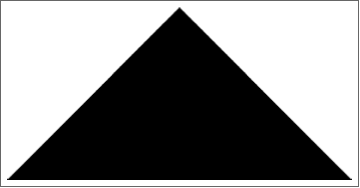}
\end{center}
\caption{Time evolution (from the top to the bottom) 
starting from a single one
of the 
ECA 18 (left) and 254 (right).}
\label{f:ECA}
\end{figure}

For  the NDECA  18, which has  $c_7=0$, $S_2=0$ and $S_1=2$, it is 
listed in Table~\ref{t:g2}, where the critical value $\lambda^*$ of 
the parameter $\lambda$ and the order parameter $a^*(\lambda)$
are reported. 
The MF prediction and numerical results are compared in 
Figure~\ref{f:18}: although in both cases the phase 
transition is observed, the quantitative match is not very good.

\begin{figure}[h]
\begin{center}
\includegraphics[height=4.6cm]{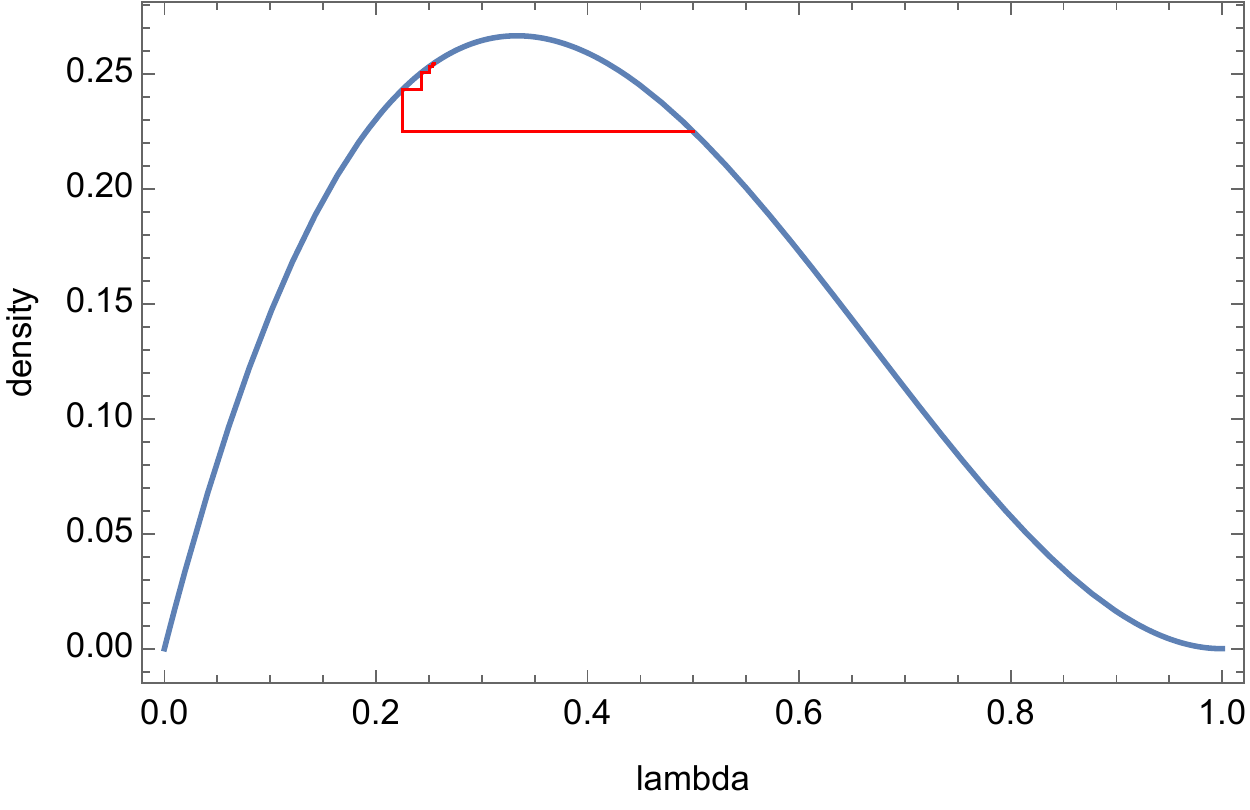}
\hskip 0.5 cm
\includegraphics[height=5.6cm]{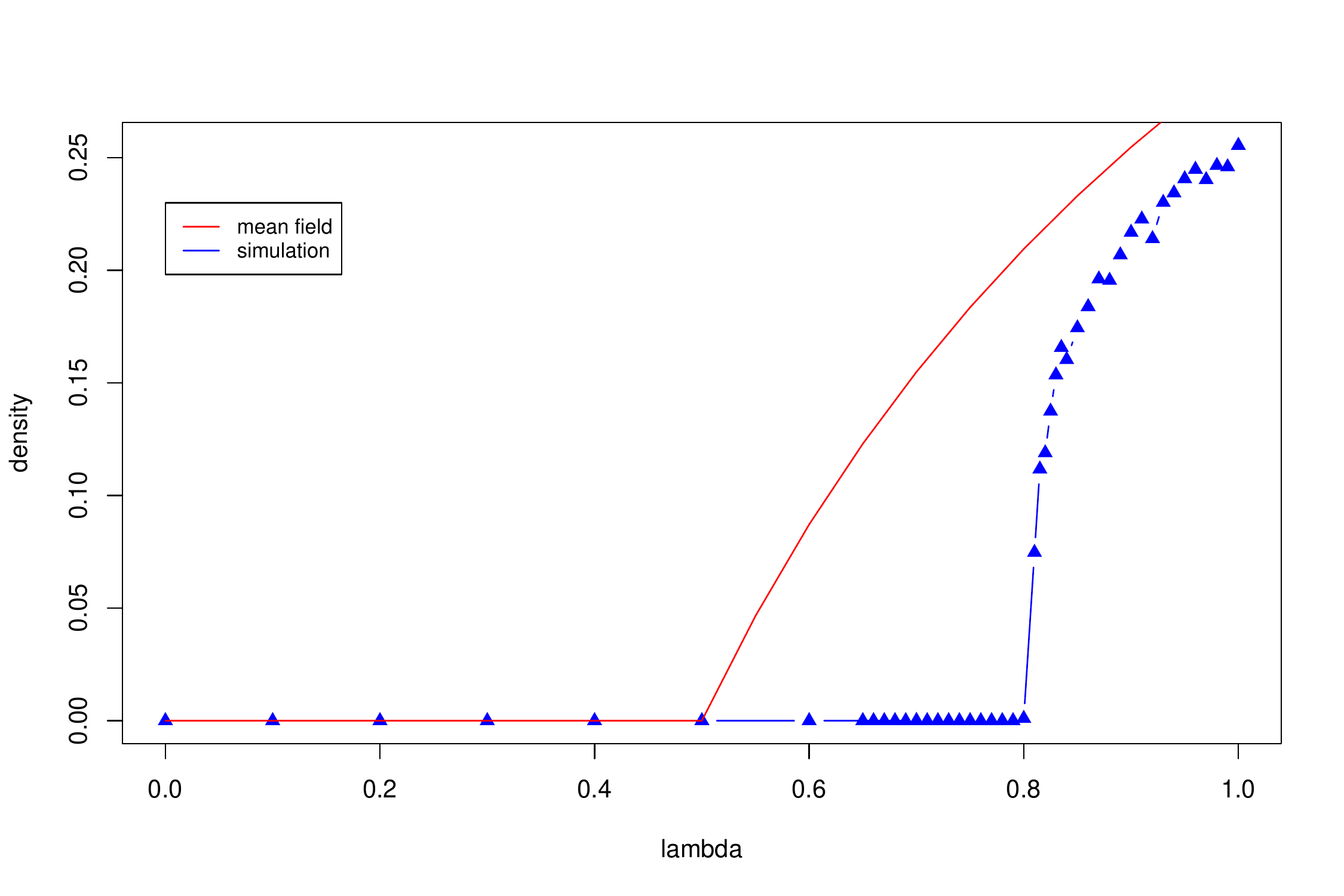}
\end{center}
\caption{(Color online)
On the left: finding the fixed point of equation \eqref{e:a_general}
for $\lambda=0.9$ for the ECA 18. On the right: comparison between 
the MF prediction (solid line) for the order parameter and 
the numerical results (dots). 
Simulations have been performed starting from an initial configuration 
with density $1/2$, on a lattice $n=10000$ and running the 
simulations for the time $5000$.}
\label{f:18}
\end{figure}

We consider now the NDECA 254(11111110). ECA 254
is a simple rule having a configuration with all ones as fixed 
point (Wolfram's class W1),  see the right panel
in Figure~\ref{f:ECA}. For this diploid we have $c_7=1$, $S_2=3$ and $S_1=3$, and the associated NDECA is 
listed in Table~\ref{t:g2}, where the critical value $\lambda^*$ of 
the parameter $\lambda$ and the order parameter $a^*(\lambda)$
are reported. 
The MF prediction and numerical results are compared in 
Figure~\ref{f:254}: 
the quantitative match is very good far from the critical point.

This diploid is well known in the literature and 
in \cite{BMM2013,T2004} is called 
\emph{percolation Probabilistic Cellular Automata}.  
In \cite[Example~2.4]{BMM2013} it is proven that there exists
$\lambda_\textup{c} \in (0,1)$ such that
for $\lambda<\lambda_\textup{c}$ the map is ergodic
and 
for $\lambda>\lambda_\textup{c}$ there are several invariant measures.
In other words, for this map  the paper \cite{BMM2013} 
provides a rigorous proof of the existence 
of the phase transition. 
The exact value of $\lambda_\textup{c}$ 
is not known, but it is proven that it
belongs to the interval $[\frac{1}{3}, \frac{53}{54}]$, see \cite{T2004}. 
A sharper result has been given in \cite{taggi}, where the lower bound $\lambda_\textup{c}>0.505$ is proven.  Therefore, for the NDECA 254 the infinite volume situation is 
close to the simulation results discussed in \cite{Fautomata2017} and illustrated in  Figure~\ref{f:254}. The simulations are indeed in an essentially infinite volume regime as we shall discuss in the sequel. We finally notice that the MF prediction $\lambda^*=1/2$ is very close to the lower bound
$\lambda_\textup{c}>0.505$.

A third example is given by NDECA 238 (11101110). This model is called \emph{Stavskaya} model and it is a particular case of the \emph{percolation} PCA when we choose $Q=\{0,1\}$ instead of $Q=\{-1,0,+1\}$.  \emph{Stavskaya} model has a phase transition for $\lambda>\lambda_\textup{c}$, with 
 $\lambda_\textup{c}>0.677$ (see \cite{MScmp1991}). Our MF approximation gives $\lambda^*=0.5$ ($c_7=1,\, S_2=3,\, S_1=2$).

Another example is the NDECA 102 (01100110), known as \emph{additive noise} PCA and it is proven to be ergodic for all $\lambda$ (Proposition~3.5 in \cite{MM2014}). This result is compatible with the MF approximation that predicts the uniqueness of the invariant measure.

As  a final example we consider the NDECA 17 (00010001), an \textit{odd NDECA}, for which the MF predicts a unique non--null
stationary measure. This model is known as \emph{directed animals} PCA (see for instance Figure~7 
in \cite{MM2014}) and it has been proven to have a unique invariant 
\emph{Markovian} measure.

\begin{figure}[h]
\begin{center}
\includegraphics[height=4.5cm]{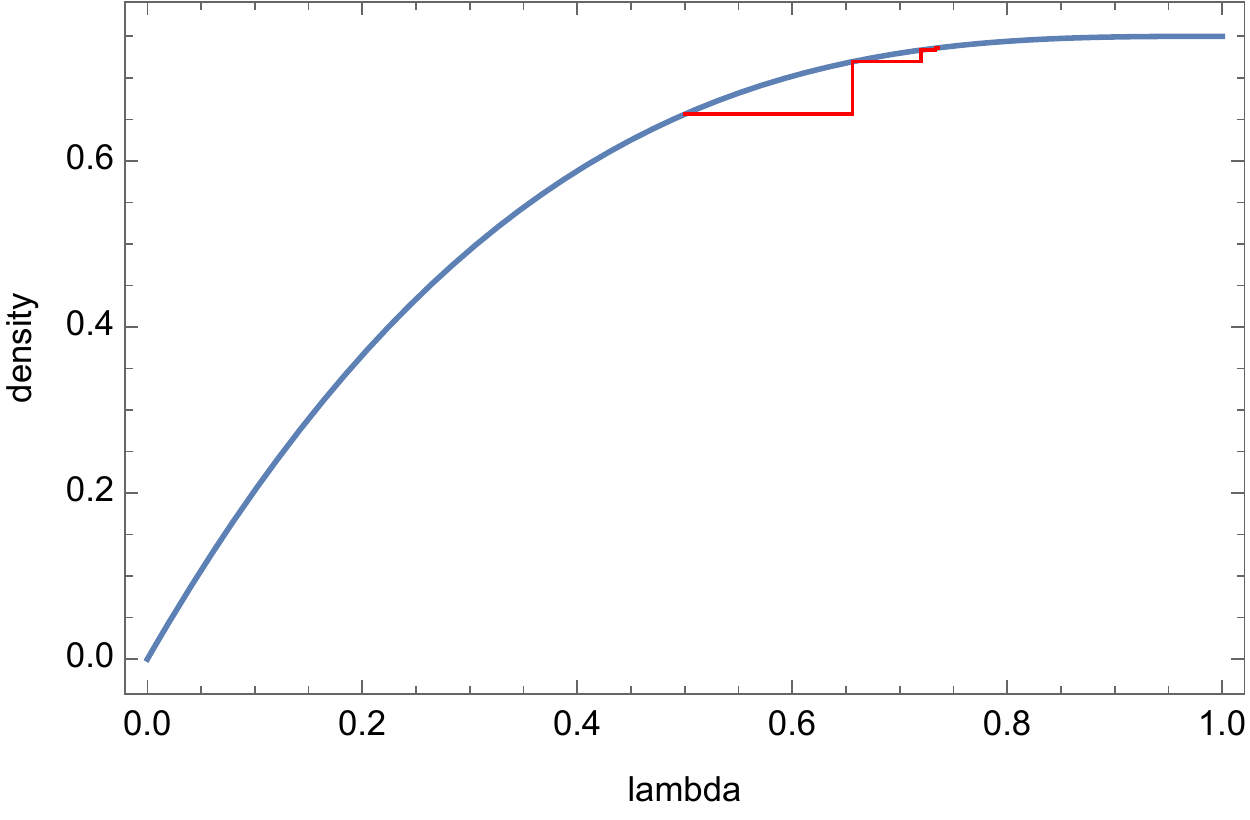}
\hskip 0.5 cm
\includegraphics[height=5.6cm]{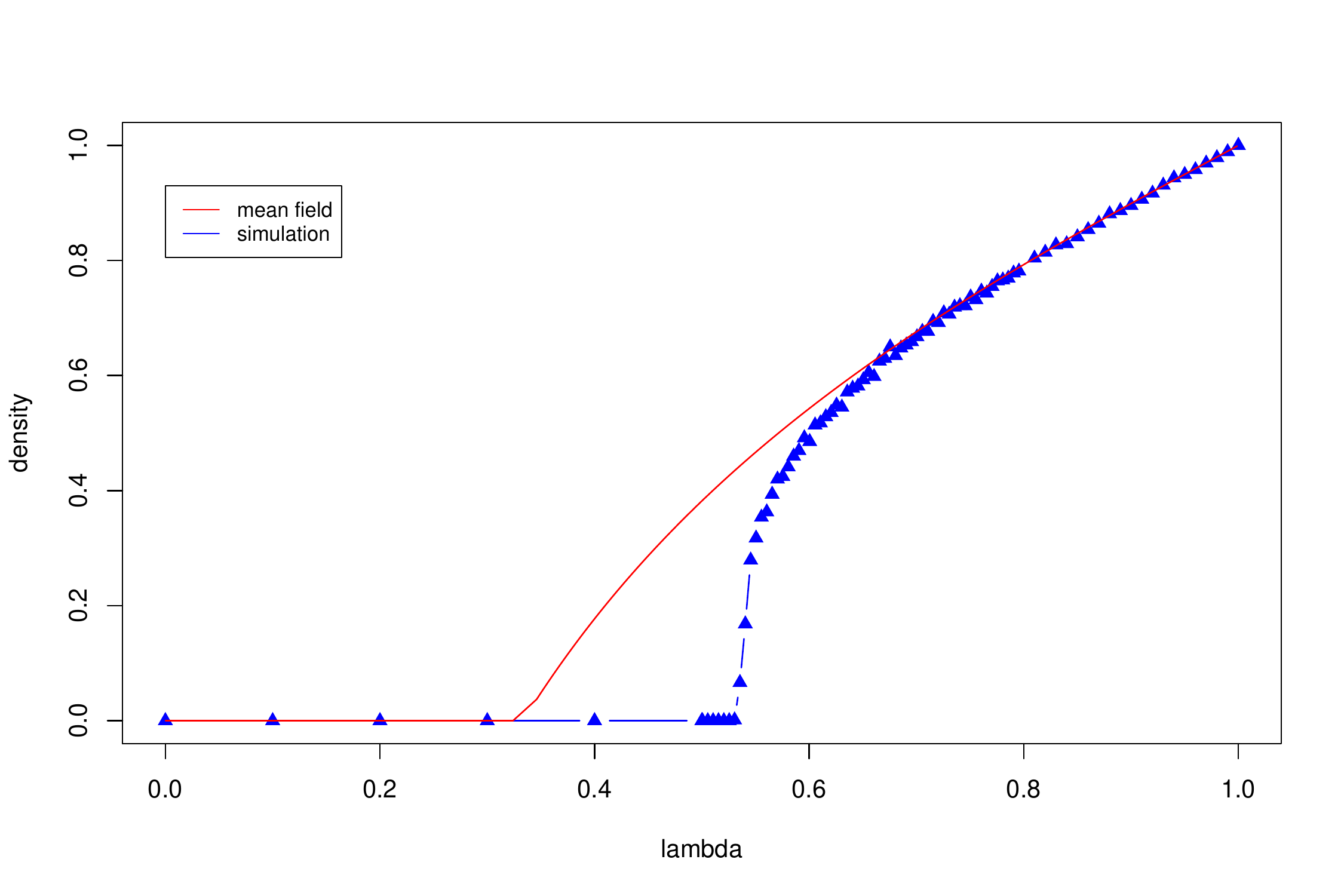}
\end{center}
\caption{(Color online)
On the left: finding the fixed point of equation \eqref{e:a_general}
for $\lambda=0.75$ for the ECA 254. On the right: comparison between 
the MF prediction (solid line) for the order parameter and 
the numerical results (dots). 
Simulations have been performed starting from an initial configuration 
with density $1/2$, on a lattice $n=10000$ and running the 
simulations for the time $5000$.}
\label{f:254}
\end{figure}

\section{Rigorous bounds for the critical point}
\label{s:infinite} 
\par\noindent
The DECA defined in Section~\ref{s:modello} with $f_1$ the null 
rule is 
considered here on $\mathbb{Z}$. 
Following \cite{MM2014}, for some finite subset $K\subset\mathbb{Z}$, consider $y =(y_k)_{k\in K}$. The \emph{cylinder} of base $K$ defined by $y$ is the set:
$$
[y]:=\left\{z\in Q^\mathbb{Z}: \forall k\in K, z_k=y_k\right\}
$$
Thus, the probability of the cylinder of base $K$ corresponding to $y$ of the chain started in $x$, can be written as:
\begin{equation}
\label{mod_infi}
P_x([y])
:=
\prod_{i\in K}
p_i(y_i|x)
\end{equation}
with 
\begin{equation}
\label{mod_infi02}
p_i(y_i|x)
=
y_i \lambda f_2(x_{i-1},x_i,x_{i+1})
+
(1-y_i)[1-\lambda f_2(x_{i-1},x_i,x_{i+1})]
.
\end{equation}
Misusing the notation, we denote again by 
$P_x$ the probability associated with the 
infinite volume process started at $x\in X=Q^{\mathbb{Z}}$
and by 
$\mu^x_t(\cdot)$ the probability measure of the 
chain started at $x$ at time 
$t$. 

In this framework 
the Dobrushin single site sufficient condition \cite{Dpit1971},
stated in \cite[equation (1--2)]{MScmp1991} in the case 
of Probabilistic Cellular Automata
and extended in \cite[Main Theorem]{MScmp1991},
provides an instrument to prove ergodicity, and hence 
existence of a unique invariant measure, for NDECA.
Let us introduce first the \emph{Dobrushin} parameter 
\begin{equation}
\label{e:d}
d=\sup_{a\in Q}
\sum_{i\in\mathbb{Z}}
\sup_{x\in Q^\mathbb{Z}}
|p_0(a|x)-p_0(a|x^i)|
\end{equation}
where $x^i$ is the configuration such that 
$(x^i)_j=x_j$ for $j\neq i$ and 
$(x^i)_i=1-x_i$. 
By using the \emph{single--site Dobrushin Criterion}
\cite{MScmp1991}, we have that 
if $d<1$ then the NDECA has a unique invariant measure.

We shall use this result, which will simply call
the \emph{Dobrushin criterion}, to find lower bounds to the 
critical value $\lambda_\textup{c}$ of the parameter $\lambda$.
Indeed, the criterion will allow us to prove uniqueness 
of the invariant measure for $\lambda$ smaller than 
some value $\mu$, which will provide a lower bound 
to $\lambda_\textup{c}$, that is to say,
$\lambda_\textup{c}\ge\mu$.

It is useful to note that, 
in our NDECA context, where we have only two symbols, the 
Dobrushin parameter simplifies to 
\begin{equation}
\label{e:d015}
d=
\sum_{i\in\mathbb{Z}}
\sup_{x\in Q^\mathbb{Z}}
|p_0(1|x)-p_0(1|x^i)|
.
\end{equation}
Moreover, using that 
$p_0(1|x)$ depends only on the value of the cells $i=-1,0,+1$, 
we can finally write
\begin{equation}
\label{e:d02}
d=
\sum_{i\in\{-1,0,+1\}}
\sup_{x\in Q^\mathbb{Z}}
|p_0(1|x)-p_0(1|x^i)|
.
\end{equation}

\begin{theorem}
\label{t:rig00}
For any choice of the ECA rule $f_2$, 
\begin{enumerate}
\item\label{i:rig00.1}
for any $i\in\{-1,0,+1\}$ and $x\in Q^3$, 
$|p_0(1|x)-p_0(1|x^i)|$ is either $0$ or $\lambda$;
\item\label{i:rig00.2}
the NDECA defined in (\ref{mod_infi}) has a unique invariant 
measure for all $\lambda<1/3$.
\end{enumerate}
\end{theorem}

\begin{proof}
Item~\ref{i:rig00.1}:
from \eqref{mod_infi02},
the probability 
$p_0(1|x)$ is either $0$ or $\lambda$, hence 
$|p_0(1|x)-p_0(1|x^i)|$ is either $0$ or $\lambda$. 
Item~\ref{i:rig00.2}:
Using item~\ref{i:rig00.1} and \eqref{e:d02}
we can then conclude that 
$d\le 3\lambda$.
Thus, for  $\lambda<1/3$ the Dobrushin criterion 
ensures that the invariant measure is unique  
in the infinite volume case and coincides with the delta measure in $0$. 
\end{proof}
 
Hence, by Theorem~\ref{t:rig00}  we have the lower bound $\lambda_\textup{c}\ge 1/3$
 for the 
critical value of the parameter\footnote{By using Theorem~3.9 in \cite{MM2014} ergodicity can be proven for $\lambda<1/3$. However, this criterion will not allow improvement when one consider subclasses of DECA. }
$\lambda$.
This estimate is rather poor if compared to the 
numerical and MF results, which predicts a value around $0.7$ 
for the critical point $\lambda_\textup{c}$.
On the other hand, 
since the result in the Theorem~\ref{t:rig00}
is uniform in the choice of the rule $f_2$ of the 
NDECA, one can expect that a better bound could be found 
if the Dobrushin criterion were applied to a particular subset of rules.
In order to realize a useful classification of the NDECA 
we introduce the following notion: 
we say that the cell 
$i\in\{-1,0,+1\}$ is \emph{marginal} 
for the NDECA if and only if 
$p_0(1|x)= p_0(1|x^i)$
for any $x\in Q^3$.
Note that, if 
$i\in\{-1,0,+1\}$ is not marginal 
for the NDECA then there exists 
$x\in Q^3$
such that 
$p_0(1|x)\not= p_0(1|x^i)$.
In other words if a NDECA has unessential cells, there exists 
a not empty subset of $\{-1,0,+1\}$ such that, for any configuration, 
the probability to set one at the origin is not affected 
if the value of a cell in this subset is varied, whereas 
there are configurations such that 
it changes if the value of any other cell is modified. 

\begin{theorem}
\label{t:rig05}
If $A\subset\{-1,0,+1\}$ is the maximal (with respect to inclusion)
set of marginal cells, then $d=(3-|A|)\lambda$.
\end{theorem}

\begin{proof}
By item~\ref{i:rig00.1} of Theorem~\ref{t:rig00} and 
the definition of marginal cells, 
we have that 
\begin{displaymath}
\sup_{x\in Q^3}|p_0(1|x)-p_0(1|x^i)|=
\left\{
\begin{array}{ll}
0 & \;\;\;\textup{if } i\in A\\
\lambda & \;\;\;\textup{otherwise}.\\
\end{array}
\right.
\end{displaymath}
The theorem then follows from \eqref{e:d02}.
\end{proof}

The above theorem allows a full classification of NDECA
with respect to the number of marginal cells. Depending on 
this number the Dobrushin parameter can be exactly computed 
for the NDECA and hopefully the estimates of the 
critical point $\lambda_\textup{c}$ can be improved. 
Note that, in particular, that for a NDECA 
not having any marginal cells, since the maximal 
set of marginal set is the empty set, 
the Dobrushin 
parameter is $3\lambda$ so that in these cases the general blind  bound 
of Theorem~\ref{t:rig00} is not improved. 
In the following sections all the possible cases will be reviewed. 

\subsection{Three marginal cells}
\label{s:mc3} 
\par\noindent
This case is rather trivial. We consider the ECA
mapping all configurations to the same cell value. 
Namely, we consider the maps 0(00000000) and 255(11111111): 
the first one maps all configurations to zero and the second 
all configurations to one. 
In both cases the number of marginal cells is three, so, by using 
Theorem~\ref{t:rig05}, we have that $d=0$. 
Hence, by the Dobrushin criterion it follows that these two 
NDECA have a single invariant measure for any $\lambda\in(0,1)$. 

\subsection{Two marginal cells}
\label{s:mc2} 
\par\noindent
All possible cases are listed in Table~\ref{t:d2}. Since the position of the 
not marginal cell can be chosen in three possible ways, 
we have the following  six choices for the $f_2$ map.
The not marginal cell is the left one:
240(11110000), 15(00001111).
The not marginal cell is the central one: 
204(11001100), 51(00110011).
The not marginal cell is the right one: 
170(10101010), 85(01010101).
It s interesting to notice that the ECA 
$240$, $204$, and $170$ have 
a straightforward interpretation in terms of shift operators: 
right shift, identity, left shift. 
For the  NDECA with rule $f_2$ one of the six rules listed above, 
since the maximal set of marginal set has cardinality equal to two, 
the Dobrushin 
parameter is $\lambda$.
Hence, by the Dobrushin criterion it follows that these six
NDECA have a single invariant measure for any $\lambda\in(0,1)$. 

This rigorous result is coherent with simulations and the MF analysis, 
indeed for the six maps listed above neither simulations nor MF 
predict the existence of the phase transition.

\begin{table}
\begin{center}
\begin{tabular}{c|c|c|c|c}
\hline\hline
not marginal & marginal & marginal & $f_2$ & $f_2$ \\
\hline\hline
1 & 1 & 1 & \multirow{4}{*}{1}& \multirow{4}{*}{0}\\
1 & 1 & 0 & & \\
1 & 0 & 1 & & \\
1 & 0 & 0 & & \\
\hline
0 & 1 & 1 & \multirow{4}{*}{0}& \multirow{4}{*}{1}\\
0 & 1 & 0 & & \\
0 & 0 & 1 & & \\
0 & 0 & 0 & & \\
\hline\hline
\end{tabular}
\end{center}
\caption{Possible choices of the rule $f_2$ in case of 
two marginal cells. First column: not marginal cell. 
Second and third columns: marginal cells. Fourth and 
fifth columns possible values of $f_2$. 
}
\label{t:d2}
\end{table}

\subsection{One marginal cell}
\label{s:mc1} 
\par\noindent
Half of the possible cases are listed in Table~\ref{t:d1}, 
the remaining one can be found as described in the 
caption. 
Since the position of the 
marginal cell can be chosen in three possible ways, 
the five cases reported in the table gives rise to the following 
fifteen choices for the $f_2$ map.
Marginal cell on the left:
153(10011001),
136(10001000),
187(10111011),
238(11101110),
221(11011101).
Marginal cell at the center:
165(10100101),
160(10100000),
175(10101111),
250(11111010),
245(11110101).
Marginal cell on the right:
195(11000011),
192(11000000),
207(11001111),
252(11111100),
243(11110011).
The remaining fifteen maps can be found by those 
reported above by changing the zeroes with the ones so that 
the complement to 255 is found. Namely, we have: 
102(01100110),
119(01110111),
68(01000100),
17(00010001),
34(00100010)
for the marginal cell on the left, 
90(01011010),
95(01011111),
80(01010000),
5(00000101),
10(00001010)
for the marginal cell as central cell, 
60(00111100),
63(00111111),
48(00110000),
2(00000011),
12(00001100)
for the marginal cell on the right.

For the  NDECA with rule $f_2$ one of the thirty rules listed above, 
since the maximal set of marginal set has cardinality equal to one, 
the Dobrushin 
parameter is $2\lambda$.
Hence, by the Dobrushin criterion it follows that these thirty NDECA
have a single invariant measure for any $\lambda<1/2$, which 
gives the lower bound $\lambda_\textup{c}\ge1/2$ for the critical point. 

It is worth noting that 
for the rules 2, 10, 34, 48, 68, 80, 136, 160, 192 neither simulations 
nor the MF analysis predict the phase transition. 
For the rules 10 and 252 the simulations do not observe the phase 
transition, whereas the MF approximation predict the existence 
of the transition with critical point $\lambda^*=1/2$.
Finally, for the rules 60, 90, 238, and 250
both simulations and the MF analysis 
predict the phase transition with a MF estimate of the 
critical point $\lambda^*=1/2$.

\begin{table}
\begin{center}
\begin{tabular}{c
|c
|c
|c
|c
|c
|c
|c
}
\hline\hline
not marginal & not marginal & marginal & 
$f_2$ & 
$f_2$ & 
$f_2$ & 
$f_2$ & 
$f_2$ 
\\
\hline\hline
1 & 1 & 1 & 
\multirow{2}{*}{1}
& 
\multirow{2}{*}{1}
& 
\multirow{2}{*}{1}
& 
\multirow{2}{*}{1}
& 
\multirow{2}{*}{1}
\\
1 & 1 & 0 & 
& 
& 
& 
& 
\\
\hline
1 & 0 & 1 & 
\multirow{2}{*}{0}
& 
\multirow{2}{*}{0}
& 
\multirow{2}{*}{0}
& 
\multirow{2}{*}{1}
& 
\multirow{2}{*}{1}
\\
1 & 0 & 0 & 
& 
& 
& 
& 
\\
\hline
0 & 1 & 1 & 
\multirow{2}{*}{0}
& 
\multirow{2}{*}{0}
& 
\multirow{2}{*}{1}
& 
\multirow{2}{*}{1}
& 
\multirow{2}{*}{0}
\\
0 & 1 & 0 & 
& 
& 
& 
& 
\\
\hline
0 & 0 & 1 & 
\multirow{2}{*}{1}
& 
\multirow{2}{*}{0}
& 
\multirow{2}{*}{1}
& 
\multirow{2}{*}{0}
& 
\multirow{2}{*}{1}
\\
0 & 0 & 0 & 
& 
& 
& 
& 
\\
\hline\hline
\end{tabular}
\end{center}
\caption{Possible choices of the rule $f_2$ in case of 
one marginal cell. 
First and second column: not marginal cells. 
Third column: marginal cell. Other columns: 
possible values of $f_2$. 
Five more cases are not listed: the values of the map $f_2$ are obtained 
by exchanging the symbols $0$ and $1$ in those reported in the table.
}
\label{t:d1}
\end{table}

\subsection{Examples}
\label{s:exe2}
\par\noindent
In this section we will look again at the examples given  Section~\ref{s:exe}, under the perspective of the rigorous results of Section~\ref{s:infinite}.

For the NDECA 18,  NDECA 17 (\emph{directed animals}) and NDECA 102 (\emph{noisy additive PCA}) we have the Dobrushin bound: $\lambda_\textup{c}\geq 1/3$.
However, in case of  NDECA 102 (\emph{noisy additive} PCA),  NDECA 238 (\emph{Stavskaya model}) and NDECA 254 ((\emph{percolation PCA} )) we have one  \emph{marginal cell} (the left), so that  $\lambda_\textup{c}>1/2$ (see Section~\ref{s:mc1}), compatible with the known results reviewed in Section~\ref{s:exe}.

\section{Time scales for finite volume diploids}
\label{s:finito} 
\par\noindent
In this section we change the perspective and examine the system in finite 
volume within time scales increasing with $n$. 
The main question is that of understanding to which extent 
finite volume simulations are a reasonable description of the 
infinite volume behaviors of diploids.

We will confine our discussion to NDECA with even $f_2$ rule, for which, 
as we have already remarked in Section~\ref{s:modello}, 
the measure concentrated on the zero configuration $0$ is an
invariant measure\footnote{In the odd case we have seen that 
both MF and simulations predict absence of phase transition
in infinite volume.
We thus expect the existence of a single invariant measure 
with not zero density. 
This can be easily proven in some simple cases. 
For instance, consider the trivial diploid where  $f_2$ is the rule $255$:
each cell is updated independently on the others and also on the 
past. Hence, the evolution of a cell is a sequence of Bernoulli variables 
with parameter $\lambda$. 
Thus, the invariant measure of the chain is product measure
and for each cell $i$ it is equal to $\pi_i(0)=1-\lambda$ and 
$\pi_i(1)=\lambda$.} 
at finite volume, since 
$p(0,x)=0$ for any $x\in X_n$. 
Moreover, since starting 
from any configuration the probability that the chain reaches the state $0$ 
is finite, we expect that any simulation, sooner or later, 
will be trapped in $0$. The aim of this section is precisely 
that of giving an estimate of the time needed by the chain 
to hit $0$. 

We start with a a very rough 
heuristic argument suggesting that for $\lambda$ small enough 
the chain should reach the configuration $0$ in a time 
logarithmically increasing with the size $n$.
Thus, consider a NDECA with even $f_2$ rule and 
choose the configuration $1$ as initial state: 
\begin{itemize}
\item[--] 
at the first step (time $1$) 
the number of $1$'s switched to $0$ is $(1-\lambda)n$, so that 
the number of ones at time $1$ is $n-(1-\lambda)n=\lambda n$.
\item[--]
At the second step the number of $1$'s turned to $0$ 
will be $\lambda n(1-\lambda)$. But at this stage 
one has to consider that zeros can be switched to one: since the rule 
$f_2$ is even, one zero, in order to have the chance to be 
turned to $1$, must at least have a $1$ among its neighboring site.
This, indeed, depends on the rule, but in this simple argument 
we consider the case which is most favorable to the $0$ to $1$ 
switch and assume that one single neighboring $1$ is sufficient 
to perform the switch according to the rule $f_2$. 
Under such assumption (exaggerating) we can estimate the number of 
zeros turning to one as twice the number of ones at time one 
times $\lambda$, namely, $2\lambda n\lambda$.
Hence, at time two the number of ones will be 
$\lambda n-(1-\lambda)\lambda n+2n\lambda^2=3\lambda^2n$. 
\item[--]
Iterating the computation at time three we find $9\lambda^3n$ ones 
and at time $t$ we will find $(3\lambda)^tn/3$ ones. 
This number will be of order one at $t\sim\log(n/3)/\log(1/(3\lambda))$, 
meaning that for $\lambda<1/3$ we expect that in a logarithmic time 
the chain will converge to the zero configuration. 
\end{itemize}
Thus, for $\lambda$ small the configuration $0$ is reached in a time 
logarithmically increasing with the size $n$. This is 
checked numerically in the left panel of Figure~\ref{f:30}
for the $f_2$ map 254. 

\begin{figure}[t]
\begin{center}
\includegraphics[height=4.5cm]{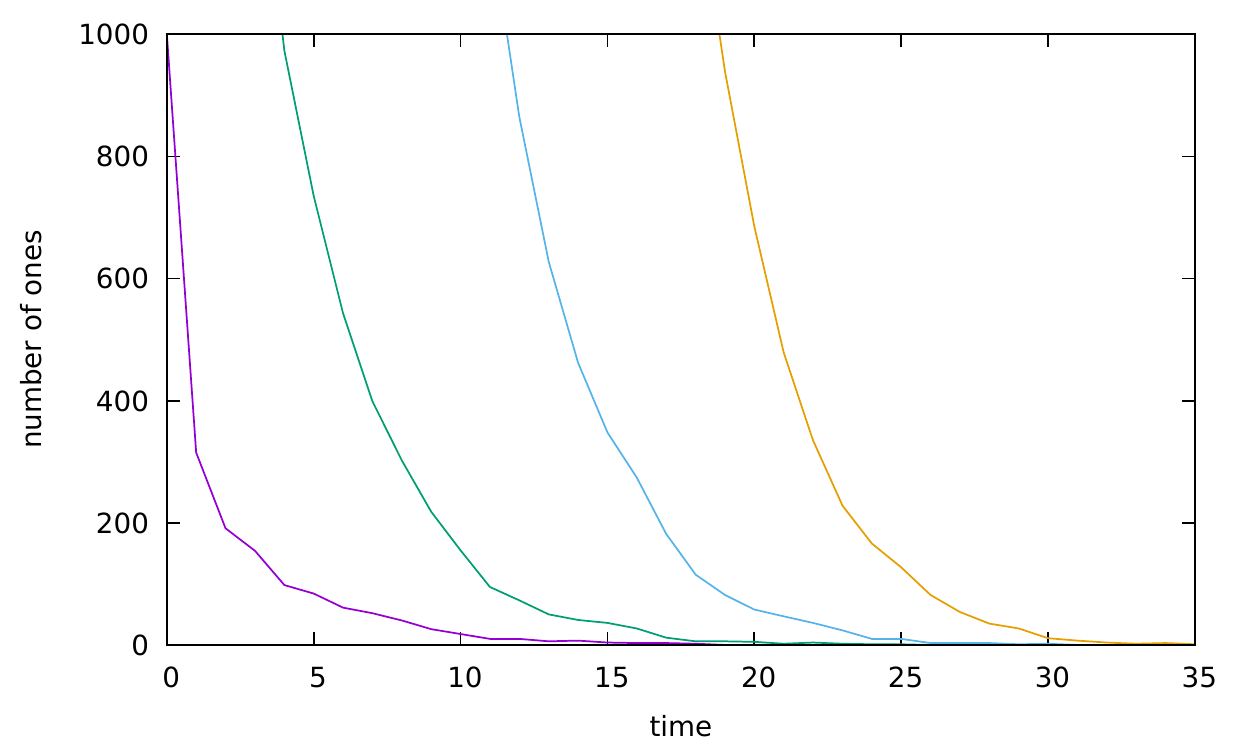}
\includegraphics[height=4.5cm]{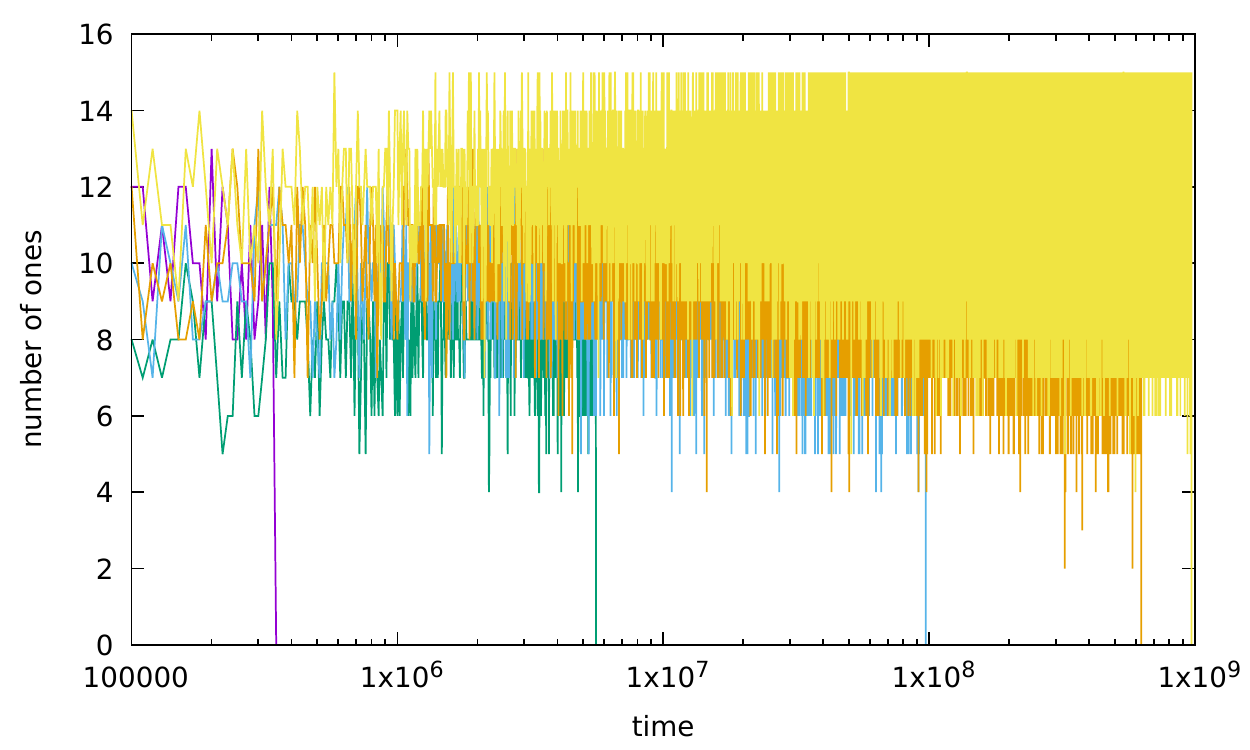}
\end{center}
\caption{(Color online)
Relaxation of NDECA with map $f_2$ 254 
starting from the 1 configuration.
On the left
$\lambda=0.30$,
$n=10^3$ (violet), 
$n=10^4$ (green), 
$n=10^5$ (blue),
$n=10^6$ (orange);
the corresponding values of the estimate $\log(n)/log(1/(3\lambda))$ 
for the relaxation time are
$55.1$, 
$77.0$, 
$98.8$,
$120.7$.
On the 
right:
the size $n$ of the lattice and the parameter $\lambda$ are chosen as 
follows: 
$n=15$ and $\lambda=0.7$ (violet), 
$n=10$ and $\lambda=0.82$ (green), 
$n=12$ and $\lambda=0.81$ (blue),
$n=13$ and $\lambda=0.81$ (orange),
$n=15$ and $\lambda=0.81$ (yellow);
the corresponding values of the estimate $1/(1-\lambda)^n$ 
for the relaxation time are
$6.9\times10^7$,
$2.8\times10^7$,
$4.5\times10^8$,
$2.4\times10^9$,
$6.6\times10^{10}$.
}
\label{f:30}
\end{figure}

The natural question, now, is about the behavior of the chain for 
$\lambda$ close to $1$. 
In the following lemma we give an upper bound 
for the  probability of being at time $t$ in a configuration different 
from the stationary state $0$ and will provide us with 
an argument to estimate the relaxation time for $\lambda$ 
close to $1$. 

\begin{theorem}
\label{l:conv}
Consider a NDECA with even map $f_2$. 
For any initial state $y\in X_n\setminus \{0\}$
we have that 
\begin{equation}
\label{e:upper}
P_y(\{\xi^t\neq0\})
\leq [1-(1-\lambda)^n]^t
\leq \exp\{-t (1-\lambda)^n\}
.
\end{equation}
\end{theorem}

\begin{proof}
Recall that the Markov chain 
is denoted by $\xi^t$. Since $0$ is a fixed point, 
for any 
initial state the event 
$\{\xi^t\in X_n\setminus\{0\}\}$ is a subset of the event 
$\{\xi^{t-1}\in X_n\setminus\{0\}\}$.
Thus, given the initial state $y$, we have 
\begin{displaymath}
P_y(\{\xi^t\neq0\})
=
P_y(\{\xi^t\neq0\}\cap\{\xi^{t-1}\neq0\})
=
\sum_{x\neq0} P_y(\{\xi^t\neq0\}\cap\{\xi^{t-1}=x\})
\end{displaymath}
and using the Markov property we get
\begin{displaymath}
P_y(\{\xi^t\neq0\})
=
\sum_{x\neq0} P_y(\{\xi^t\neq0\}|\{\xi^{t-1}=x\})
P_y(\{\xi^{t-1}=x\})
\end{displaymath}
which yields the recursive bound 
\begin{displaymath}
\begin{array}{rcl}
P_y(\{\xi^t\neq0\})
&\le&
{\displaystyle
\sup_{x\neq0} P_y(\{\xi^t\neq0\}|\{\xi^{t-1}=x\})
\sum_{x\neq0} 
P_y(\{\xi^{t-1}=x\})
}
\\
&=&
{\displaystyle
\sup_{x\neq0} P_y(\{\xi^t\neq0\}|\{\xi^{t-1}=x\})
P_y(\{\xi^{t-1}\neq0\})
\;\;.
}
\end{array}
\end{displaymath}
Moreover, 
we note that
\begin{displaymath}
\begin{array}{rcl}
{\displaystyle
\sup_{x\neq0} P_y(\{\xi^t\neq0\}|\{\xi^{t-1}=x\})
}
&=&
{\displaystyle
\sup_{x\neq0} [1-P_y(\{\xi^t=0\}|\{\xi^{t-1}=x\})]
}
\\
&=&
{\displaystyle
1-\inf_{x\neq0} P_y(\{\xi^t=0\}|\{\xi^{t-1}=x\})
.
}
\end{array}
\end{displaymath}
Since to put $0$ in a cell has a probability cost at least  
$1-\lambda$, we have that
\begin{displaymath}
\inf_{x\neq0} P_y(\{\xi^t=0\}|\{\xi^{t-1}=x\})\ge (1-\lambda)^n
\;\;.
\end{displaymath}
Collecting all the bounds and iterating with respect to $t$ we have that 
\begin{displaymath}
\begin{array}{rcl}
P_y(\{\xi^t\neq0\})
&\le&
{\displaystyle
[1-(1-\lambda)^n]
P_y(\{\xi^{t-1}\neq0\})
}
\\
&\le&
{\displaystyle
[1-(1-\lambda)^n]^{t-1}
P_y(\{\xi^1\neq0\})
\le
[1-(1-\lambda)^n]^t
.
}
\end{array}
\end{displaymath}
The second bound is immediate.
\end{proof}

As we noticed before for $\lambda$ small a time $t(n)$
diverging logarithmically with $n$ 
seems to be sufficient for the finite volume diploid to approach
the $0$ state, in the sense that 
$P_y(\{\xi^{t(n)}\neq0\})$ tends to zero as $n\to\infty$. 
The above theorem, for any 
$\lambda\in(0,1)$, proves a weaker, but rigorous, statement: a time 
$t(n)$
diverging exponentially with $n$ is 
sufficient for the finite volume diploid to relax to 
the $0$ state in the sense specified above. Indeed, 
if $t(n)=\alpha^{n}$, with $\alpha\ge 1/(1-\lambda)$, then from 
Theorem~\ref{l:conv} it follows that 
\begin{displaymath}
P_y(\{\xi^{t(n)}\neq0\})
\leq e^{-[\alpha(1-\lambda)]^n}
\to0
\end{displaymath}
in the limit $n\to\infty$.
The Theorem~\ref{l:conv} is useless for times $t(n)$ smaller than 
$(1-\lambda)^n$, namely, such that 
$t(n)(1-\lambda)^n\to0$. Indeed, in such a case the r.h.s.\ 
of \eqref{e:upper} tends to $1$ and the bound is trivial.

This behavior, which shares some common feature with metastable 
states, is checked 
numerically in the right panel of Figure~\ref{f:30} for the 
NDECA with $f_2$ map number 254.
We had to use 
ridiculously small lattices, i.e., $n=10,12,13,15$, 
due to the exponential dependence of the relaxation time on $n$.
In the picture on the horizontal axis we report the time on a 
logarithmic scale and on the vertical axis we report the number of cells 
with value one.
Notice that in a time of order $10^9$ all the diploid ECA except the 
yellow one relax to the stationary state $0$. 

Finally, we come back to the original question about the 
ability of the simulation to catch the infinite volume behavior 
of the NDECA.
As it is absolutely reasonable,
from the Theorem~\ref{l:conv} it follows immediately that 
for a fixed $n$ the diploid converges to the 
stationary state $0$ with probability one
in the limit $t\to\infty$.
How is this result compatible with the numerical studies presented in  
Section~\ref{s:modello}? Indeed, 
those simulations are obviously performed at finite volume, 
nevertheless the system is found in a stationary state with 
density different from zero. 
The key is the choice of the time--scales considered in the simulations 
and the size $n$ of the chain:
in the simulations $n=10^4$ and $t=5\cdot10^3$, 
the bound (\ref{e:upper}) is thus irrelevant, indeed, 
$[1-(1-\lambda)^n]^t=[1-0.15^{10^4}]^{5 \cdot10^3}\cong 1$.
It is reasonable to suppose that for that choice of the parameters,  
the system is essentially in the \emph{infinite volume} regime, 
where the probability of flipping to zero at the same time a large number 
of cells is negligible.

The problem of the relaxation time has been treated at a high level 
of generality, namely, we considered any NDECA with $f_2$ an even map. 
In this perspective
it is not possible to be more precise about the behavior of the 
relaxation time with respect to the volume $n$ of the system. 
On the other hand, considering particular NDECA one can prove 
more precise statements, as we discuss in the following subsections.

\subsection{The case of the identity: logarithmic behavior of 
the relaxation time}
\label{s:tc-204}
\par\noindent
Consider the NDECA 
with the map $f_2$ being the identity, namely, the rule 204(11001100): 
as mentioned in Section~\ref{s:modello}
this rule associates to any configuration 
the state of the cell at the center of the neighborhood, namely, the 
cell that one is going to update. 
Cells are thus updated independently one from each other but not on 
the past. The chain can be described as a collection of $n$ 
single cell Markov chains evolving with the transition matrix
\begin{displaymath}
q(0,0)=1,\;
q(0,1)=0,\;
q(1,0)=1-\lambda,\;
q(1,1)=\lambda,
\end{displaymath}
for any cell $i\in\mathbb{L}_n$.
The stationary measure is product and the single 
cell stationary measure is concentrated on $0$, namely, 
$\pi_i(0)=1$ and $\pi_i(1)=0$. 
Moreover, for a single cell started at $1$, the probability 
that its state is $0$ at time $t$ is equal to $1-\lambda^t$, namely, $1$ minus 
the probability that from time $1$ to time $t$ it has always been 
sampled the rule $f_2$. 
Hence, if we look at the whole chain, 
we have\footnote{This model can be solved also by using multinomial 
distributions. One can sum over all the ways in which 
$1$'s are removed. If $s_k$ is the number of ones removed at time 
$k$ we have 
\begin{displaymath}
\mu_t^1(0)
=
\sum_{s_1+\cdots+s_t=n}
\left(\newatop{n}{s_1}\right)
\left(\newatop{n-s_1}{s_2}\right)
\cdots
\left(\newatop{n-s_1-\cdots-s_{t-1}}{s_t}\right)
(1-\lambda)^{s_1}\lambda^{n-s_1}
\cdots
(1-\lambda)^{s_t}\lambda^{n-s_1-\cdots-s_t}
.
\end{displaymath}
Expanding the binomials and 
distributing the $\lambda$ and $1-\lambda$ terms we get 
\begin{displaymath}
\mu_t^1(0)
=
\lambda^{nt}
\!\!\!\!\!\!
\sum_{s_1+\cdots+s_t=n}
\!\!
\frac{n!}{s_1!\cdots s_t!}
(1-\lambda)^{\sum_{k=1}^ts_k}
\lambda^{-t s_1}
\lambda^{-(t-1) s_2}
\cdots
\lambda^{-s_t}
=
\lambda^{nt}
\!\!\!\!\!\!
\sum_{s_1+\cdots+s_t=n}
\!\!
\frac{n!}{s_1!\cdots s_t!}
\prod_{k=1}^t
\Big(
    \frac{1-\lambda}{\lambda^{t-k+1}}
\Big)^{s_k}
\end{displaymath}
Exploiting the multinomial theorem we get
\begin{displaymath}
\mu_t^1(0)
=
\lambda^{nt}
\bigg(
\sum_{k=1}^t
     \frac{1-\lambda}{\lambda^{t-k+1}}
\bigg)^n
=
\frac{\lambda^{nt}(1-\lambda)^n}{\lambda^{n(t+1)}}
\bigg(
\sum_{k=1}^t
     \lambda^k
\bigg)^n
=
\frac{(1-\lambda)^n}{\lambda^n}
\Big(
     \frac{1-\lambda^{t+1}}{1-\lambda}-1
\Big)^n
=(1-\lambda^t)^n
.
\end{displaymath}
}
that 
$\mu_t^1(0)=(1-\lambda^t)^n$. 
Now, suppose to compute this probability on a time scale 
diverging logarithmically with $n$, namely, take $t(n)=\alpha\log n$
for some $\alpha$ such that $\alpha>-1/\log\lambda$:
\begin{displaymath}
\mu_t^1(0)
=
(1-\lambda^t)^n
\approx
e^{-n\lambda^{\alpha\log n}}
=
e^{-\exp\{(1+\alpha\log\lambda)\log n\}}
\to0
\end{displaymath}
in the limit $n\to\infty$.
Thus,
for any $\lambda\in(0,1)$, the NDECA under consideration will 
converge in probability to $0$, namely in this case the 
relaxation time is logarithmically large in $n$ (see
Figure~\ref{f:20}).

\begin{figure}[t]
\begin{center}
\includegraphics[height=4.5cm]{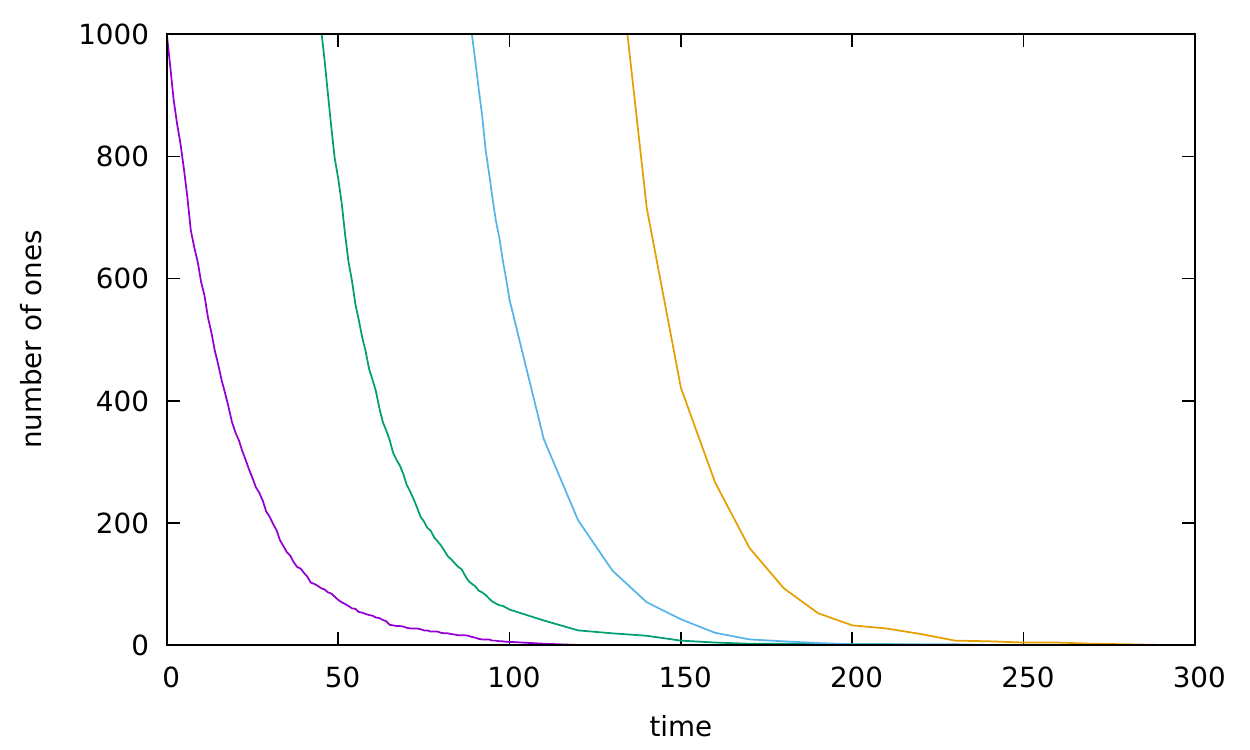}
\end{center}
\caption{(Color online) 
Relaxation of NDECA with map $f_2$ 204 (identity) 
starting form the 1 configuration
with
$\lambda=0.95$,
$n=10^3$ (violet), 
$n=10^4$ (green), 
$n=10^5$ (blue),
$n=10^6$ (orange);
the corresponding estimates $(-1/\log\lambda)\log n$ for the relaxation time
are
$134.7$, 
$179.5$, 
$224.4$,
$269.3$.
}
\label{f:20}
\end{figure}

\subsection{The case of the percolation PCA: exponential behavior 
of the relaxation time}
\label{s:tc-254}
\par\noindent
We consider, here, the NDECA with map 254 as map $f_2$, which, 
as discussed above, is an example of percolation PCA. 
In Figure~\ref{f:30} we have shown that for $\lambda$ small 
the relaxation time diverges logarithmically with $n$, whereas 
for $\lambda$ large this divergence is exponential. 
We have supported these conclusions with some analytical argument. 

Indeed, for the percolation PCA this result is proven rigorously 
in \cite[Theorem~2.1]{taggi}. In this paper the author
proves for the critical point the bound $\lambda_{\textup{c}}\ge 0.505$,
see the table in the Appendix therein. Moreover the 
Theorem~2.1 which provides an estimate for the average 
relaxation time, can be restated as follows.

\begin{theorem}
\label{t:taggi}
For the NDECA with map $f_2$ the ECA 254, 
let $\tau$ be the first hitting time to the state $0$ starting from the 
state $1$, then there exists $n_0\in\mathbb{N}$ and 
some positive constants 
$K_1$, $K_2$, $K_3$, $K_4$, $c_1$, $c_2$, $c_3$, and $c_4$
(dependent on $\lambda$) such that for all $n>n_0$
\begin{itemize}
\item[--]
if $\lambda<\lambda_{\textup{c}}$ then 
$K_1\log(c_1n)\le\mathbb{E}_1[\tau]\le K_2\log(c_2n)$;
\item[--]
if $\lambda>\lambda_{\textup{c}}$ then 
$K_3\,\textup{exp}(c_3n)\le\mathbb{E}_1[\tau]\le K_4\,\textup{exp}(c_4n)$;
\end{itemize}
where we denote by $\mathbb{E}_1$ the mean on the trajectories 
of the Markov chain started at $1$.
\end{theorem}

\section{Conclusions}
\label{s:conclusioni} 
\par\noindent
We have studied the possibility for diploid Elementary Cellular 
Automata (DECA) to exhibit phase transitions. 
In particular we have analyzed the case in which one of the two 
ECA mixed to obtain the DECA is the null rule. In such case we have 
called NDECA the DECA so obtained.

The problem has been approached via a MF approximation and 
through the use of the rigorous Dobrushin Criterion. 
The two methods have allowed two different classifications 
of NDECA. The two approaches give coherent results and, to some 
extent, explain and justify some of the numerical 
results discussed in \cite{Fautomata2017}.

As we have often repeated, the point of view followed in this paper, 
and mainly borrowed from \cite{Fautomata2017}, allow a 
unified approach to many different PCA which has been studied in the 
Probability and Statistical Mechanics literature 
putting them in a different light. In some dedicated sections, for 
the PCA that we have been able to spot in the past literature, 
we have compared our results with some rigorous statements 
already present in the literature.

\bigskip
\bigskip
\par\noindent
\textbf{Acknowledgements}
\par\noindent
The authors thank R.\ Fernandez and N. Fat\`es  for very useful discussions. 
ENMC expresses his thanks to the Mathematics Department of the 
Utrecht University for kind hospitality and STAR for financial support. The research of Francesca R.~Nardi was partially supported by the NWO Gravitation Grant 024.002.003--NETWORKS and by the PRIN Grant 20155PAWZB  \emph{Large Scale Random Structures}.


\end{document}